\documentclass[accepted]{uai2026} 
                        
\usepackage{graphicx} 
\usepackage{natbib}  
\usepackage{multirow}

\usepackage[utf8]{inputenc}
\usepackage[ruled]{algorithm}
\usepackage[noend]{algpseudocode}
\usepackage{hyperref}
\usepackage{mathtools}
\usepackage{booktabs}
\usepackage[T1]{fontenc}
\usepackage{nicefrac}
\usepackage{microtype}
\usepackage[pdftex,dvipsnames]{xcolor}
\usepackage{caption}
\usepackage{subcaption}
\usepackage{amsfonts, amsmath, amssymb, amsthm}
\usepackage[nameinlink,capitalize]{cleveref}
\usepackage{enumitem}

\usepackage{bbm}
\usepackage{bm}    

\usepackage[suppress]{color-edits}
\addauthor{dn}{blue}
\addauthor{vp}{olive}
\addauthor{so}{magenta}

\newcommand{\declarecolor}[2]{\definecolor{#1}{RGB}{#2}\expandafter\newcommand\csname #1\endcsname[1]{\textcolor{#1}{##1}}}

\declarecolor{White}{255, 255, 255}
\declarecolor{Black}{0, 0, 0}
\declarecolor{Maroon}{128, 0, 0}
\declarecolor{Coral}{255, 127. 80}
\declarecolor{Red}{182, 21, 21}
\declarecolor{LimeGreen}{50, 205, 50}
\declarecolor{DarkGreen}{0, 90, 0}
\declarecolor{Navy}{0, 0, 128}

\definecolor{plotblue}{HTML}{377eb8}
\definecolor{plotorange}{HTML}{ff7f00}
\definecolor{plotgreen}{HTML}{4daf4a}

\hypersetup{
   colorlinks=true,
   citecolor=DarkGreen,
   linkcolor=Navy,
   filecolor=magenta,
   urlcolor=blue,
   pdfpagemode=FullScreen,
}

\usepackage{subcaption}
\usepackage[thinc]{esdiff}
\usepackage{physics}

\usepackage{authblk}

\usepackage{tablefootnote}
\usepackage{xspace}
\usepackage{adjustbox}

\Crefname{figure}{Figure}{Figure}

\newtheorem{theorem}{Theorem}[section]
\newtheorem*{theorem*}{Theorem}
\newtheorem{lemma}[theorem]{Lemma}
\newtheorem*{lemma*}{Lemma}

\newtheorem{remark}[theorem]{Remark}
\newtheorem{corollary}[theorem]{Corollary}

\newtheorem{assumption}[theorem]{Assumption}

\newtheorem{problem}[theorem]{Problem}

\newtheoremstyle{named}{}{}{\itshape}{}{\bfseries}{.}{.5em}{\thmnote{#3 }#1}
\theoremstyle{named}

\makeatletter
\def\namedlabel#1#2{\begingroup
   \def\@currentlabel{#2}%
   \label{#1}\endgroup
}
\makeatother

\newcommand{\unif}{\mathsf{Unif}}

\newcommand{\cA}{\mathcal{A}}

\newcommand{\cN}{\mathcal{N}}

\newcommand{\cW}{\mathcal{W}}
\newcommand{\cG}{\mathcal{G}}

\newcommand{\bX}{\mathbf{X}}

\newcommand{\eps}{\epsilon}

\newcommand{\E}{\mathbb E}

\newcommand{\R}{\mathbb R}

\newcommand{\1}{\mathbf{1}}

\newcommand{\ie}{{i.e.,~\xspace}}
\newcommand{\eg}{{e.g.,~\xspace}}

\newcommand{\trunc}{\mathrm{truncate}}

\newcommand{\xtr}{x_\mathrm{tr}}

\newcommand{\floor}[1]{\left\lfloor#1\right\rfloor}

\newcommand{\bv}[1]{\mathbf{#1}}

\newcommand{\Xb}{\mathbf{X}}

\newcommand{\sgn}{\textnormal{sgn}}

\newcount\Comments
\newcommand{\kibitz}[2]{\ifnum\Comments=1\textcolor{#1}{#2}\fi}

\newcommand{\mcenter}[1]{\mathrm{center}(#1)}

\usepackage{tikz}

\renewcommand{\algorithmicrequire}{\textbf{Input:}}
\renewcommand{\algorithmicensure}{\textbf{Output:}}

\usepackage{listings}

\usepackage[american]{babel}

\usepackage{natbib}
\usepackage{mathtools} 
\usepackage{booktabs} 
\usepackage{tikz} 
\usetikzlibrary{arrows.meta, positioning, shapes.geometric, calc}
\usepackage{adjustbox}

\makeatletter
\newcommand{\equalcontrib}{$\dagger$}
\newcommand{\equalcontribfootnote}[1]{
    \begingroup
        \renewcommand{\thefootnote}{$\dagger$}
        \footnotetext{#1}
    \endgroup
}
\makeatother

\makeatletter
\newcommand{\other}{*}
\newcommand{\otherfootnote}[1]{
    \begingroup
        \renewcommand{\thefootnote}{\other}
        \footnotetext{#1}
    \endgroup
}
\makeatother

\title{Adaptive and Robust Watermark for Generative Tabular Data}

\author[1, \equalcontrib]{Dung Daniel Ngo}
\author[2, \equalcontrib]{Archan Ray}
\author[2, \equalcontrib]{Akshay Seshadri}
\author[1, \other]{Daniel Scott}
\author[1]{Saheed Obitayo}
\author[2]{Niraj Kumar}
\author[1]{Vamsi K. Potluru}
\author[2]{Marco Pistoia}
\author[1]{Manuela Veloso}

\affil[1]{AI Research, JPMorganChase, New York, NY 10017, USA } 
\affil[2]{Global Technology Applied Research, JPMorganChase, New York, NY 10001, USA}

\begin{document}
\maketitle
\equalcontribfootnote{Equal Contribution}
\otherfootnote{Work done while at AI Research, JPMorganChase}
\begin{abstract}
  In recent years, watermarking generative tabular data has become a prominent framework to protect against the misuse of synthetic data. However, while most prior work in watermarking methods for tabular data demonstrate a wide variety of desirable properties (\eg high fidelity, detectability, robustness), the findings often emphasize empirical guarantees against common oblivious and adversarial attacks. In this paper, we study a flexible and robust watermarking algorithm for generative tabular data. Specifically, we demonstrate theoretical guarantees on the performance of the algorithm on metrics like fidelity, detectability, robustness, and hardness of decoding. The proof techniques introduced in this work may be of independent interest and may find applicability in other areas of machine learning. Finally, we validate our theoretical findings on synthetic and real-world tabular datasets.
\end{abstract}

\section{Introduction}
In the age of increasingly sophisticated generative models \citep{chatgpt, gemmateam2024gemmaopenmodelsbased, zhang2022opt, jiang2023mistral7b, grattafiori2024llama3herdmodels}, the use of synthetic data has become ubiquitous in areas where often real data is scarce (e.g., two examples) or under intense regulatory guidelines (e.g., health care~\citep{gonzales2023synthetic} and finance~\citep{assefa2020generating,potluru2023synthetic}). This is partly because generating synthetic data is cost-efficient and, importantly, eliminates the need for human involvement~\citep{monarch2021human, fui2023generative}. However, there is a growing concern that carelessly adopting synthetic data with the same frequency as human-generated data may lead to misinformation and privacy breaches~\citep{carlini2021extracting, stadler2022synthetic}. Furthermore, quality-control in synthetic data is important as repeatedly training models on low-quality AI-generated data exhibit gradual degradation in performance \citep{shumailov2024ai} (also called `model collapse'). For instance, in medical setting, many hospitals with fewer patients do not have enough internal data to train their rare disease prediction model. While using synthetic data to augment their internal model is a promising direction, repeatedly generating and relying on synthetic patient data as the primary source of data bring forth both ethical and performance concerns. Hence, they often rely on data sharing initiatives, where a shared platform have aggregated real patient data from other large hospitals as well as industry firms and academic institutes. Here, one of the data sharing platform's goal is to ensure that the downstream small hospital can verify that this data source comes from real patient data and not synthetic data. From this data taken from the shared platform, the smaller hospitals can then train their prediction models without worrying about the authenticity of their data source. 

To this end, watermarking generative data has gained prominence~\citep{kirchenbauer2023watermark, kuditipudi2024robust, an2024benchmarking, he2024watermarkinggenerativetabulardata, zheng2024tabularmarkwatermarkingtabulardatasets}. Here, a hidden pattern is embedded in the data that may be indiscernible to an oblivious human, yet can be detected through an efficient procedure. Several measures are used to gauge the quality of such algorithms including fidelity, detectability, and robustness~\citep{katzenbeisser2000digital,atallah2001natural}. Watermarking sequential data is well-studied, but algorithms for tabular data have only recently gained attention. This is partly because tabular datasets often follow relational algebra, and are therefore harder to watermark. While prior works~\citep{he2024watermarkinggenerativetabulardata, zheng2024tabularmarkwatermarkingtabulardatasets} have proposed watermarking for tabular data, a comprehensive theoretical framework has remained elusive. 

\begin{figure*}[ht]
\centering
\begin{adjustbox}{max width=\textwidth,center}
\begin{tikzpicture}[
    font=\sffamily,
    stepblue/.style={
        rectangle, rounded corners=2pt, draw=blue!35, line width=0.5pt,
        fill=blue!5, minimum width=2.7cm, minimum height=2.7cm
    },
    steporange/.style={
        rectangle, rounded corners=2pt, draw=orange!55, line width=0.5pt,
        fill=orange!12, minimum width=2.7cm, minimum height=2.7cm
    },
    stepgreen/.style={
        rectangle, rounded corners=2pt, draw=green!45!black!45, line width=0.5pt,
        fill=green!8, minimum width=2.7cm, minimum height=2.7cm
    },
    arrow/.style={-{Stealth[length=2.2mm]}, line width=1.5pt, darkgray},
    title/.style={font=\sffamily\bfseries\footnotesize, text=black!80,
        anchor=north, align=center},
    sub/.style={font=\sffamily\scriptsize, text=black!55, align=center,
        text width=2.4cm, anchor=south}
]

\def\dx{3.1}
\def\rowA{2.6}
\def\rowB{-1.3}

\node[title, text=black!80] at (7.5,4.6) {(a) Watermark embedding pipeline (\Cref{alg:tabular})};

\node[stepblue]   (a1) at (0*\dx,\rowA) {};
\node[stepblue]   (a2) at (1*\dx,\rowA) {};
\node[stepblue]   (a3) at (2*\dx,\rowA) {};
\node[steporange] (a4) at (3*\dx,\rowA) {};
\node[steporange] (a5) at (4*\dx,\rowA) {};
\node[stepgreen]  (a6) at (5*\dx,\rowA) {};

\node[title] at ([yshift=-0.1cm]a1.north) {1. Input};
\node[title] at ([yshift=-0.1cm]a2.north) {2. Pair Columns};
\node[title] at ([yshift=-0.1cm]a3.north) {3. Bin \& Seed};
\node[title] at ([yshift=-0.1cm]a4.north) {4. Red/Green Split};
\node[title] at ([yshift=-0.1cm]a5.north) {5. Shift to Green};
\node[title] at ([yshift=-0.1cm]a6.north) {6. Output};

\node[sub] at ([yshift=0.14cm]a1.south) {Table $\mathbf X$, bins $b$, PAIR routine};
\node[sub] at ([yshift=0.14cm]a2.south) {$n$ (key, value) column pairs};
\node[sub] at ([yshift=0.14cm]a3.south) {Hash bin centers $\rightarrow$ RNG seed};
\node[sub] at ([yshift=0.14cm]a4.south) {$b$ intervals, $p=1/2$ each};
\node[sub] at ([yshift=0.14cm]a5.south) {Move red values to nearest green interval};
\node[sub] at ([yshift=0.14cm]a6.south) {Watermarked $\mathbf X_w$ + pair seeds};

\begin{scope}[shift={([yshift=0.1cm]a1.center)}]
    \draw[black!55,line width=0.5pt] (-0.4,-0.4) rectangle (0.4,0.4);
    \draw[black!55,line width=0.5pt] (0,-0.4)--(0,0.4);
    \draw[black!55,line width=0.5pt] (-0.4,0)--(0.4,0);
\end{scope}
\begin{scope}[shift={([yshift=0.1cm]a2.center)}]
    \fill[blue!30] (-0.28,-0.4) rectangle (-0.06,0.35);
    \fill[purple!30] (0.06,-0.4) rectangle (0.28,0.35);
    \draw[black!45,-{Stealth[length=1.5mm]}] (-0.17,0.5) .. controls (0,0.75) .. (0.17,0.5);
\end{scope}
\begin{scope}[shift={([yshift=0.1cm]a3.center)}]
    \foreach \x in {-0.35,-0.175,...,0.35}{\draw[black!55,line width=0.9pt] (\x,-0.3)--(\x,0.35);}
    \draw[BurntOrange,line width=1.2pt] (0,-0.3)--(0,0.35);
    \fill[BurntOrange] (0.001,-0.42) circle (0.07);
\end{scope}
\begin{scope}[shift={([yshift=0.1cm]a4.center)}]
    \fill[BrickRed] (-0.42,-0.15) rectangle (-0.18,0.15);
    \fill[green!55!black] (-0.10,-0.15) rectangle (0.14,0.15);
    \fill[BrickRed] (0.22,-0.15) rectangle (0.46,0.15);
\end{scope}
\begin{scope}[shift={([yshift=0.1cm]a5.center)}]
    \fill[BrickRed] (-0.42,-0.12) rectangle (-0.18,0.15);
    \draw[black!55,-{Stealth[length=1.6mm]}] (-0.12,0.02)--(0.10,0.02);
    \fill[green!55!black] (0.16,-0.12) rectangle (0.40,0.15);
\end{scope}
\begin{scope}[shift={([yshift=0.1cm]a6.center)}]
    \draw[black!55,line width=0.5pt] (-0.4,-0.4) rectangle (0.4,0.4);
    \draw[black!55,line width=0.5pt] (0,-0.4)--(0,0.4);
    \draw[black!55,line width=0.5pt] (-0.4,0)--(0.4,0);
    \draw[green!45!black,line width=1pt] (0.12,-0.05)--(0.24,-0.2)--(0.5,0.2);
\end{scope}

\foreach \i/\j in {1/2,2/3,3/4,4/5,5/6}{\draw[arrow] (a\i.east) -- (a\j.west);}

\draw[dashed, black!35] (-1.5,0.9) -- (16.7,0.9);

\node[title, text=black!80] at (7.5,0.7) {(b) Watermark detection pipeline (\Cref{sec:dectection})};

\node[stepblue]   (b1) at (0*\dx,\rowB) {};
\node[stepblue]   (b2) at (1*\dx,\rowB) {};
\node[stepblue]   (b3) at (2*\dx,\rowB) {};
\node[steporange] (b4) at (3*\dx,\rowB) {};
\node[steporange] (b5) at (4*\dx,\rowB) {};
\node[stepgreen]  (b6) at (5*\dx,\rowB) {};

\node[title] at ([yshift=-0.1cm]b1.north) {1. Query Table};
\node[title] at ([yshift=-0.1cm]b2.north) {2. Rebuild Labels};
\node[title] at ([yshift=-0.1cm]b3.north) {3. Count Matches};
\node[title] at ([yshift=-0.1cm]b4.north) {4. z-score};
\node[title] at ([yshift=-0.1cm]b5.north) {5. Threshold Test};
\node[title] at ([yshift=-0.1cm]b6.north) {6. Decision};

\node[sub] at ([yshift=0.14cm]b1.south) {Table $\mathbf X$ + stored pairs, seeds};
\node[sub] at ([yshift=0.14cm]b2.south) {Regenerate red/green from key column};
\node[sub] at ([yshift=0.14cm]b3.south) {$T_i$ = cells of col. $i$ in green intervals};
\node[sub] at ([yshift=0.14cm]b4.south) {$z_i=2\sqrt{m}$ $(T_i/m-0.5)$};
\node[sub] at ([yshift=0.14cm]b5.south) {Compare $z_i$ to $z_{th}$ at level $\alpha / n$};
\node[sub] at ([yshift=0.14cm]b6.south) {Watermarked if all $n$ tests reject $H_0$};

\begin{scope}[shift={([yshift=0.1cm]b1.center)}]
    \draw[black!55,line width=0.5pt] (-0.4,-0.4) rectangle (0.4,0.4);
    \draw[black!55,line width=0.5pt] (0,-0.4)--(0,0.4);
    \draw[black!55,line width=0.5pt] (-0.4,0)--(0.4,0);
    
    \draw[black!90,line width=0.6pt] (0.25,-0.25) circle (0.22);
    \draw[black!90,line width=1.2pt] (0.40,-0.40)--(0.55,-0.55);
\end{scope}
\begin{scope}[shift={([yshift=0.1cm]b2.center)}]
    \draw[black!60, line width=0.8pt] (-0.32,0) circle (0.19);
    \draw[black!60, line width=0.8pt] (-0.12,0) -- (0.45,0);
    \draw[black!60, line width=0.8pt] (0.28,0) -- (0.28,-0.14);
    \draw[black!60, line width=0.8pt] (0.40,0) -- (0.40,-0.11);
\end{scope}
\begin{scope}[shift={([yshift=0.1cm]b3.center)}]
    \fill[black!50!black!70] (-0.36,-0.3) rectangle (-0.20,0.05);
    \fill[green!50!black!70] (-0.14,-0.3) rectangle (0.02,0.30);
    \fill[black!50!black!70] (0.08,-0.3) rectangle (0.24,-0.05);
    \fill[black!50!black!70] (0.30,-0.3) rectangle (0.46,0.18);
\end{scope}
\begin{scope}[shift={([yshift=0.1cm]b4.center)}]
    \draw[black!80,line width=0.7pt] plot[domain=-0.5:0.5,samples=40]
        (\x,{0.5*exp(-(\x*\x)/0.05)-0.25});
    \draw[black!40,line width=0.4pt] (-0.5,-0.25)--(0.5,-0.25);
\end{scope}
\begin{scope}[shift={([yshift=0.1cm]b5.center)}]
    \draw[black!45,dashed,line width=0.5pt] (-0.45,-0.15)--(0.45,-0.15);
    \foreach \x in {-0.38,-0.14,0.10}{\fill[black!55] (\x,-0.15);}
    \fill[OliveGreen!80] (0.25,0.02) circle (0.055);
    \fill[Maroon!80] (-0.24,-0.32) circle (0.055);
\end{scope}
\begin{scope}[shift={([yshift=0.1cm]b6.center)}]
    \draw[green!45!black,line width=1.4pt] (-0.25,0.0)--(-0.05,-0.22)--(0.32,0.28);
\end{scope}

\foreach \i/\j in {1/2,2/3,3/4,4/5,5/6}{\draw[arrow] (b\i.east) -- (b\j.west);}

\end{tikzpicture}
\end{adjustbox}
\caption{\textbf{Watermark embedding and detection pipelines}.}
\label{fig:pipelines}
\end{figure*}

\subsection{Our Contributions}
\paragraph{Theoretical framework.} We present a comprehensive theoretical study of a simple watermarking algorithm for continuous tabular data with several desirable properties. (1) \textbf{Fidelity:} We demonstrate that data that is watermarked according to our algorithm is similar to the input data. Moreover, we can control how close they are. (2) \textbf{Detection:} We show that there is a simple algorithm to detect our watermark. (3) \textbf{Robustness:} We show that our watermarking algorithm is robust to arbitrary additive noise, and modifications due to downstream tasks such as truncation and feature selection. (4) \textbf{Decoding:} We present the first study of algorithms capable of decoding the query function used to determine whether data has been watermarked.

\paragraph{Empirical analysis.} Our second contribution is a systematic empirical study of several watermarking algorithms, including our proposed variant. Our results support theoretical findings and show that our method matches~\citet{he2024watermarkinggenerativetabulardata} in fidelity and utility, but is more robust to additive noise, while theirs is stronger against table-level attacks. Compared to~\citet{zheng2024tabularmarkwatermarkingtabulardatasets}, our approach offers higher fidelity and cell-level robustness, but lower downstream utility. Against~\citet{fangrintaw}, our method is more robust to table-level attacks, though with reduced utility.

\begin{figure*}[t]
    \centering
    \begin{adjustbox}{width=0.9\textwidth}
        \begin{tikzpicture}
        \node at (1.5,2.35) {
        \begin{tabular}{|c|c|c|c|}
            \hline
            \text{Key A} & \text{Value B} & \text{Value A} & \text{Key B} \\
            \hline
            & & & \\
            $K_1$ & \ldots & $V_1$ & \ldots \\
            & & & \\
            \hline
            & & & \\
            $K_2$ & \ldots & $V_2$ & \ldots \\
            & & & \\
            \hline
            & & & \\
            $K_3$ & \ldots & $V_3$ & \ldots \\
            & & & \\
            \hline
        \end{tabular}
        };
        
        \draw[ultra thick] (6,4.2) -- (14,4.2);
        \foreach \x in {6, 7.6, 9.2, 10.8, 12.4, 14} {
            \draw (\x,4.1) -- (\x,4.3);
        }
        \node at (6,4.6) {0.0};
        \node at (7.6,4.6) {0.2};
        \node at (9.2,4.6) {0.4};
        \node at (10.8,4.6) {0.6};
        \node at (12.4,4.6) {0.8};
        \node at (14,4.6) {1.0};
        
        \draw[fill=red] (6,3.1) rectangle (7,3.7);
        \draw[fill=red] (7,3.1) rectangle (8,3.7);
        \draw[fill=red] (8,3.1) rectangle (9,3.7);
        \draw[fill=green] (9,3.1) rectangle (10,3.7);
        \draw[fill=green] (10,3.1) rectangle (11,3.7);
        \draw[fill=green] (11,3.1) rectangle (12,3.7);
        \draw[fill=red] (12,3.1) rectangle (13,3.7);
        \draw[fill=green] (13,3.1) rectangle (14,3.7);
        \node at (11.5,3.4) {$V_1$};
        
        \draw[fill=red] (6,1.9) rectangle (7,2.5);
        \draw[fill=green] (7,1.9) rectangle (8,2.5);
        \draw[fill=green] (8,1.9) rectangle (9,2.5);
        \draw[fill=red] (9,1.9) rectangle (10,2.5);
        \draw[fill=green] (10,1.9) rectangle (11,2.5);
        \draw[fill=green] (11,1.9) rectangle (12,2.5);
        \draw[fill=red] (12,1.9) rectangle (13,2.5);
        \draw[fill=green] (13,1.9) rectangle (14,2.5);
        \node at (12.5,2.2) {$V_2$};
        
        \draw[fill=green] (6,0.7) rectangle (7,1.3);
        \draw[fill=green] (7,0.7) rectangle (8,1.3);
        \draw[fill=red] (8,0.7) rectangle (9,1.3);
        \draw[fill=red] (9,0.7) rectangle (10,1.3);
        \draw[fill=red] (10,0.7) rectangle (11,1.3);
        \draw[fill=red] (11,0.7) rectangle (12,1.3);
        \draw[fill=green] (12,0.7) rectangle (13,1.3);
        \draw[fill=green] (13,0.7) rectangle (14,1.3);
        \node at (9.5,1.0) {$V_3$};
        
        \node at (5.6,3.4) {$K_1$};
        \node at (5.6,2.2) {$K_2$};
        \node at (5.6,1.0) {$K_3$};

        \node (A) at (12.5,2.4) {};
        \node (B) at (11.5,2.4) {};
        \draw[->, ultra thick] (A) to[out=120, in=60] (B);
        \node (C) at (9.5,0.8) {};
        \node (D) at (7.5,0.8) {};
        \draw[->, ultra thick] (C) to[out=-120, in=-50] (D);        
        \end{tikzpicture}
    \end{adjustbox}
    \caption{Illustrative example of \Cref{alg:tabular} on a tabular dataset with $3$ rows and $4$ columns. This structure corresponds to $2$ pairs of $(key, value)$ columns. In the first row, the element $V_1$ is already in a `green' interval. Meanwhile, the other elements, $V_2$ and $V_3$, have to be moved from `red' interval to a nearby `green' interval.}
    \label{fig:illustrative-example}
\end{figure*}

\section{Related Work}
\label{sec:related-work}
Watermarking relational databases has been extensively studied in the literature~\citep{agrawal2002watermarking, sion2003rights, shehab2007watermarking, lin2021watermark, li2023watermark, kamran2013watermark, hwang2020reversible, hu2019reversible}, with methods classified as either non-reversible (embedding bits into data or statistics)~\citep{agrawal2002watermarking, Hamadou2011AWS} or reversible (allowing recovery of original data)~\citep{hu2019reversible, hwang2020reversible}. These works mainly focus on minimizing changes to the mean and variance of the data, while ignoring how the watermarked dataset performs for downstream machine learning tasks \citep{kamran2018comprehensivesurveywatermarkingrelational}.

Recent advances in watermarking generative tabular data has been inspired from the watermarking techniques for large language models~\citep{kamaruddin2018review, aaronson2023openai, kuditipudi2024robust, kirchenbauer2023watermark}, using either post-process or generation-time watermarks \citep{fangrintaw}. Generation-time methods, such as those in \citet{zhu2025tabwak}, show strong robustness but rely on access to the sampling process, which is often impractical.

Inspired by \citet{kirchenbauer2023watermark}'s idea for generating sequences from specific partitions in the probability space, \citep{he2024watermarkinggenerativetabulardata, zheng2024tabularmarkwatermarkingtabulardatasets} embeds the watermark by forcing the values in generated columns to specific regions of the real number line. Particularly, WGTD~\citep{he2024watermarkinggenerativetabulardata} proposed a binning scheme with neighboring `red' and `green' intervals to ensure high data fidelity and robustness against additive noise attack. TabularMark~\citep{zheng2024tabularmarkwatermarkingtabulardatasets} focused on embedding the watermark only in the prediction target feature for both regression and classification task. Importantly, these algorithms are simple, allowing us to study them theoretically as our proposed algorithm derives from this class of watermarking algorithms.
Additional details on related work is discussed in \cref{appendix:additional-related-work}.

\section{Notation and Algorithms}
\label{sec:watermark}

In this section, we first state the notation used throughout. We also state the main algorithm that we study in our work and provide a brief overview of the algorithm.

\subsection{Notation}
For $n \in \mathbb{N}^{+}$, we write $[n]$ to denote $\{1,\cdots,n \}$. For a matrix $\bX \in \mathbb{R}^{m \times 2n}$, we denote $\ell_\infty$-norm of $\bX$ as $\norm{\bX}_\infty = \max_{i \in [m], j \in [2n]} \abs{\bX_{i,j}}$. For an interval $g = [a, b]$, we denote the center of $g$ as $\mcenter{g} = \nicefrac{(a + b)}{2}$. For any $x \in \mathbb{R}$, $\lfloor x \rfloor$ is the largest integer that is less than or equal to $x$, and $x^\circ = x - \lfloor x \rfloor$ denotes its fractional part. For a tabular dataset $\bX \in [0,1]^{m \times 2n}$, $\bX_i$ is the $i$-th column of $\bX$. We denote $\bX_w$ as the watermarked table. Define $\operatorname{append}(x, Y)$ as the function that appends $x$ to an ordered list $Y$.
\subsection{A simple algorithm for watermarking.}

\begin{algorithm}[ht]
    \caption{PAIR}\label{alg:pair}
    \begin{algorithmic}[1]
        \Require Probability vector $p \in \R^{2n}$, seed $s$.
        \State Initialize RNG with seed $s$.
        \State Initialize $\operatorname{Pairs} = []$, $S \gets [2n]$.
        \State Create a random permutation $\sigma$ of $[2n]$ using $p$. \label{algline:pair-permutation} 
        \For{$i \in [n]$}
            \State $\operatorname{Pairs} = \operatorname{append}([(\sigma_{2i-1}, \sigma_{2i})], \operatorname{Pairs})$.
        \EndFor
        \Ensure $\operatorname{Pairs}$.
    \end{algorithmic}
\end{algorithm}

\begin{algorithm}[ht]
    \caption{Pairwise Tabular Watermarking}\label{alg:tabular}
    \begin{algorithmic}[1]
        \Require Tabular dataset $\bX \in \mathbb{R}^{m\times 2n}$, number of bins $b \in \mathbb{N}^+$, pairing subroutine $\mathrm{PAIR}$ (e.g., \cref{alg:pair}), hash function HASH.
        \State Construct $n$ $(key, value)$ pairs using $\mathrm{PAIR}$. 
        \State Divide the $key$ columns into bins of width $\nicefrac{1}{b}$ to form intervals denoted $\{I_j\}_{j \in [b]}\ \lvert\ I_j = [\nicefrac{j-1}{b}, \nicefrac{j}{b}]$.
        \For{each $key$ column, $\bv k$}
            \State Initialize $c = []$. For each $i \in [m]$, $c \gets \operatorname{append}(\mcenter{I_j}, c)$ where $\bv k_i^\circ \in I_j$. 
            \State $seed \gets \operatorname{HASH}($c$)$. Set RNG with $seed$.
            \State $c \sim \unif(\{\text{red}, \text{green}\})^b$, $G\gets \{i: c_i = \text{green}\}$.
            \For{each $x \in value$ column}
                \State $G^* \gets \arg\min_{g \in G} \abs{x^\circ - \mcenter{I_g}}$.
                \If{$x^\circ \notin \cup_{g \in G} I_g$}
                    \State $x \gets \lfloor x \rfloor + x_w$ with $x_w \sim \unif\left(\cup_{g \in G^*} I_g\right)$. \label{alg-line:embed-watermark}
                \EndIf
            \EndFor
        \EndFor
        \Ensure Watermarked dataset $\bX_w$, list of $(key, value)$ columns.
    \end{algorithmic}
\end{algorithm}

In this work, we study \cref{alg:tabular} for watermarking tabular data. At a high level, \cref{alg:tabular} partitions the feature space into pairs of $(key, value)$ columns using a subroutine $\mathrm{PAIR}$ (see \cref{fig:illustrative-example}). There are several ways to implement PAIR, we give one such in \cref{alg:pair}. We finely divide the range of elements in each $key$ column into bins of size $\nicefrac{1}{b}$ to form $b$ consecutive intervals. The center of the intervals for each $key$ column is used to compute a hash, which becomes the seed for a random number generator (RNG). The corresponding $value$ column is then partitioned into intervals of size $\nicefrac{1}{b}$, and the RNG is used to label these intervals as red or green with equal probability. The watermark is embedded into a $value$ column of $\bX$ by looking at the fractional part of the input data in the column, and choosing a green interval that is closest to this fractional part. This process is repeated until all the $value$ columns are watermarked. Note, in practice, it is not necessary to return entire the list of $(key, value)$ pairs. It suffices to store the parameters of $\mathrm{PAIR}$ (e.g., seed) that is used to generate the specific instance of $(key, value)$ pairs that are created while running \Cref{alg:tabular}. In what follows, we analyze several key properties of \Cref{alg:tabular}. We start by stating a key assumption in our analysis.

\begin{assumption}
    We assume that (i) the red and green intervals are labeled independently with probability $\nicefrac{1}{2}$, and (ii) different data points/samples (rows) are independent. We note that the independence across samples is a standard assumption made for simplifying the statistical analysis in machine learning literature, while the assumption that red/green labels are chosen independently is a feature of the algorithm. Importantly, we do not make any assumption on the independence of columns.
\end{assumption}

\paragraph{Notes on the analysis of \cref{alg:tabular}.} While a real-world tabular dataset may contain many categorical features, embedding the watermark in these features may cause significant changes in the meaning for the entire row. Given a dataset $\bX$ with mixed data types, one may construct a subset of the dataset consisting of only continuous data. This subset can then be watermarked and reincorporated within the original tabular data.
Given that the watermarking due to \cref{alg:tabular} is only on the fractional part of the data, throughout our theoretical analyses, we assume that the $\bX \in [0, 1]^{m \times 2n}$. Finally, to avoid trivial situations, we assume that at least one of $\min(m, n) > 1$.
%
\section{Properties of Pairwise Tabular Watermarking Algorithm}
In this section, we show that the watermarking algorithm (\Cref{alg:tabular}) satisfies several desirable properties. Specifically, in \Cref{sec:fidelity}, we show that the watermarked data is close to the original data. In \Cref{sec:dectection}, we show that there is a principled method for detecting (with high probability) whether a given tabular data is watermarked according to \Cref{alg:tabular}. In \Cref{sec:robustness}, we show that our watermarking algorithm is robust to many types of noise, as well as common downstream tasks such as feature selection and truncation. Finally, in \Cref{sec:decoding}, we provide explicit sample/query complexity upper and lower bounds for the task of decoding the watermarking scheme.
\subsection{Fidelity}
\label{sec:fidelity}
A desirable property of a watermarking algorithm is that the watermark does not change the data significantly, or in other words, maintains a `high fidelity' with the original data. We show that \cref{alg:tabular} satisfies such a property in the sense that the watermarked data is close to the true data (with respect to entry-wise $\ell_\infty$-distance) with high probability over the choice of red/green labeling of the intervals and choice of the fractional part of the data for watermarking.
\begin{theorem}[Fidelity]\label{thm:fidelity}
Let $\bX \in [0,1]^{m \times 2n}$ be a tabular dataset, and $\bX_w$ is its watermarked version from \cref{alg:tabular}. With probability at least $1 - \delta$, for $\delta \in (0,1)$, the entry-wise $\ell_\infty$-distance between $\bX$ and $\bX_w$ is bounded above by:

\begin{equation}
    \norm{\bX - \bX_w}_\infty \leq \min  \left \{ \frac{1}{2b}\left(\log_2\left(\frac{mn}{\delta}\right) +1\right), 1 \right \} 
\end{equation}
\end{theorem}
Observe that as the number of bins increase, the watermarked data is guaranteed to be close to the input data with high probability. Intuitively, this is because (1) a large distance between an element $x$ of the input tabular data and the corresponding element $x_w$ in the watermarked data implies that there are many red-labeled intervals between $x$ and $x_w$, and (2) it is very unlikely that a large number of contiguous intervals are assigned red (since red and green are chosen with equal probability). 
Furthermore, \cref{thm:fidelity} yields a corollary that upper bounds the Wasserstein distance between the empirical distributions of $\bX$ and $\bX_w$. Together, these results show that $\bX_w$ is, in expectation, close to $\bX$, so downstream tasks on $\bX_w$ incur only $O(\nicefrac{1}{b})$ additional error with high probability. 
\begin{corollary}[Wasserstein distance]
Let $F_{\bX} = \sum_{j = 1}^m \frac{1}{m} \delta_{\bX[j, :]}$ be the empirical distribution built on $\bX \in [0,1]^{m \times 2n}$, and $F_{\bX_w} = \sum_{j=1}^m \frac{1}{m} \delta_{\bX_w[j,:]}$ be the empirical distribution built on $\bX_w$. With probability at least $(1 - \delta) \in (0,1)$, the $k$-Wasserstein distance is upper bounded by
\begin{equation}
    \cW_k (F_{\bX}, F_{\bX_w}) \leq \frac{\sqrt{n/2}}{b} \left(\log_2\left(\frac{mn}{\delta}\right)+ 1\right).
\end{equation}
\label{cor:wasserstein}
\end{corollary}

\subsection{Detection}
\label{sec:dectection}
Suppose that the data provider watermarked some input tabular data and allowed others to use it. At a later time, an external party claims that the data being used by them was supplied by the data provider. To check the veracity of such claims, we devise a procedure that can `detect' whether a given tabular data has been watermarked according to \cref{alg:tabular}.

Recall that given access to the $(key, value)$ columns, the data provider can use the $key$ columns to reconstruct the red/green intervals for each $value$ column. If any data falls in red bins, it is likely not watermarked. However, due to noise or other computational factors, some elements in $\bX_w$ may fall outside green bins. To address this, we formulate detection as a hypothesis test and characterize the null distribution, enabling robust detection against minor data modifications.
\paragraph{Hypothesis test}
\begin{align*}
    &H_0: \text{Dataset } X \text{ is not watermarked} \\
    &H_{0,i}: \text{The $i$-th
    $value$ column is not watermarked} \\
    &H_1: \text{Dataset } X \text{ is watermarked}
\end{align*}
That is, when the null hypothesis holds, all of the individual null hypotheses for the $i$-th value column must hold simultaneously. Thus, the data provider who wants to detect the watermark for a dataset $\bX$ would need to perform the hypothesis test for each $value$ column individually. If the goal is to reject the null hypothesis $H_0$ when the $p$-value is less than a predetermined significant threshold $\alpha$ (typically $0.05$ to represent $5\%$ risk of incorrectly rejecting the null hypothesis), then the data provider would check if the $p$-value for each individual null hypothesis $H_{0,i}$ is lower than $\alpha/n$ (after accounting for the family-wise error rate using Bonferroni 
correction for testing $n$ hypotheses~\citep{bonferroni1936teoria}).

An important step towards performing such a hypothesis test is to understand the null distribution, \ie a baseline on how the input data is distributed in the red/green intervals without performing any watermarking (lines $7$-$11$ in \cref{alg:tabular}). The following lemma gives us such a baseline.
\begin{lemma}
\label{lem:convergence}
    Let $F$ be the probability distribution of the data $x$ which is to be watermarked (a single element of a value column of the input table) with support in $[0,1]$. Define $\cG$ as the union of intervals which are labeled green. Then, we have $\Pr[x \in \cG] = 1/2$ where the probability is over $x\sim F$ and over the choice of the labels of the intervals (red/green).
\end{lemma}
Notably, \cref{lem:convergence} states that irrespective of the distribution from which the elements of the input data are sampled from, the probability that they fall in a green bin is $\nicefrac{1}{2}$ (over random choice of red and green bins as well as the data distribution). Therefore, if $T_i$ denotes the number of elements in the $i$-th $value$ column that falls into a green interval, then, under the individual null hypothesis $H_{0,i}$, we know that $T_i \sim B(m, \nicefrac{1}{2})$. Using the Central Limit Theorem, we have $2\sqrt{m} \left( \frac{T_i}{m} - \frac{1}{2} \right) \rightarrow \cN(0,1)$ as $m \to \infty$.
Hence, we use the statistic for a one-proportion $z$-test for testing each individual null hypothesis, which is
\begin{equation}
    z_i = 2 \sqrt{m} \left(\frac{T_i}{m} - \frac{1}{2} \right).
    \label{eq:z-score}
\end{equation}
For a given pair of $(key, value)$ columns, one can calculate the corresponding $z$-score by counting the number of elements in $value$ column that are in green intervals. 
Given a significance level $\alpha$, $z_{\text{th}}$, the threshold for the $z$-score, can be computed and used to reject each null hypothesis (at significance level $\alpha/n$) using the quantile function of the standard normal distribution. 

\subsection{Robustness}
\label{sec:robustness}
In practice, data that has been watermarked can get perturbed due to several reasons. These could be due to addition of noise to the data for data processing, or because the user decides to select only a few features from the large tabular data that is relevant to their application. We call changes of this form `oblivious attacks' since they can affect the statistical watermark embedded in the data, but are not done with the intention of changing the watermark. On the other hand, a malicious user or an `adversary' can try to modify the data while still maintaining the watermark. Changes of this type are called `adversarial attacks'. In this section, we focus on robustness against oblivious attacks.

The main robustness result that we prove in this study is a bound on the number of entries of the data that one needs to corrupt in order to break the watermark. Interestingly, we have both a distribution dependent as well as a distribution independent bound for the case of additive noise, which we present below.

\begin{theorem*}[\ref{thm:robustness}, Informal]
    Fix a column in the watermarked dataset. Let $n_a$ be the number of cells in this column that one can inject (additive) noise into. Suppose that the noise injected into these cells is drawn i.i.d. from some distribution.
    Then, there is a number $\gamma \in [0, 1]$ depending on the noise distribution such that
    $
        n_a \geq \frac{m-z_{\textnormal{th}}\sqrt{m}}{2-\gamma},
    $
    for the expected $z$-score to be less than $z_{\textnormal{th}}$ (i.e., to remove the watermark).
\end{theorem*}

A direct consequence is that one needs to perturb at least $(m - z_{\textnormal{th}} \sqrt{m})/2$ cells (\cref{rem:robustness-independent-of-distro}) to ensure that the watermark is `broken', no matter what the noise distribution is. Intuitively, the entity injecting noise does not know what intervals are labeled red/green. Therefore, when the noise is injected `blindly', the perturbed elements can not only fall into a red interval but also into a green interval. Thus, to remove the watermark, one needs to add noise into sufficiently many cells so that enough cells fall into a red interval.

We can use \cref{thm:robustness} to get explicit bounds for specific distributions such as uniform or Gaussian distribution (see \cref{appendix:robustness} for details). Similarly, we also prove that our watermarking scheme is robust under other types of tasks that may be performed in data science applications such as feature selection or truncation of elements (to save memory). Proofs of robustness for these tasks is given in \cref{appendix:robustness}. 

\paragraph{Truncation.} Often, due to compute resources, $\bX_w$ can be truncated. We define truncation of $x \in \bX_w$ as 
\begin{equation}
    \xtr = \trunc(x,p) = \frac{\floor{10^p \cdot x}}{10^p}.
    \label{eq:trunc}
\end{equation}
In \cref{thm:truncation}, we show that the probability that the watermarked value under truncation falls in a red bin increases with $b$ (and inversely to $1/b$). This result presents an interesting trade-off between smaller bin width for higher fidelity (see \Cref{thm:fidelity}) and bigger bin width for better robustness, which has not been studied in prior work on watermarking tabular data. 
\begin{theorem}
    Let $p$ be the precision of a given watermarked element $x \in I_j = [\nicefrac{j-1}{b}, \nicefrac{j}{b}]$ truncated using~\Cref{eq:trunc}. The probability that the truncated element $\xtr$ falls out of its original green interval is 
    \begin{equation*}
        \Pr [\xtr \notin I_j] = \frac{(b-1)^{10^p} + b^{10^p-1}(c \cdot b - j + 1)}{b^{10^p}} 
    \end{equation*}
    where $c \in \{0, 0.01, \cdots, 0.99 \}$ are the left grid points of $I_j$ intervals. 
\label{thm:truncation}
\end{theorem}

\paragraph{Feature selection.} Feature selection is column sampling procedure which is often of relevance in data science~\citep{cherepanova2023performance, covert2023learning, xu2023efficient}. Since the watermarks are done on pairs of columns of the original table, to detect watermark after feature selection, some of the $key, value$ pairs would be required to be retained. As such, we study the problem of robustness to feature selection. Specifically, given the importance of each feature columns are known a priori, \cref{lem:robustness-prob-bound} demonstrates that if PAIR uses feature importance to construct pairs, then the probability that a pair is retained post feature selection is greater than the probability of preserving pairs constructed uniformly at random.
\begin{lemma*}[\ref{lem:robustness-prob-bound}, Informal]
    The probability of sampling a $(key, value)$ pair from the top $k$ important columns according to $p$, is greater when sampling without replacement according to $p$ than when sampling without replacement uniformly.
\end{lemma*}
Intuitively \cref{lem:robustness-prob-bound} holds since the problem of pairing boils down to how the permutation of $[2n]$ is constructed (using sampling without replacement). For additional results on feature selection, see \cref{appendix:feature-selection}.
\subsection{Decoding watermark}
\label{sec:decoding}
For any watermarking scheme, it is important that it is secure against spoofing. Specifically, we do not want an adversary to decode the watermarking scheme so that they can make changes to it without the data provider's knowledge.
In this regard, it is important to understand how many queries are required by an adversary to learn a watermarking scheme. In turn, this would enable us to bound the total number of queries we allow a single user to make to the watermarked tabular data before we run the risk of giving out the exact parameters of our algorithms. Although such a result has not been obtained in prior work on watermarking generative tabular data, we believe it would showcase the power of any watermarking scheme that is similar to \cref{alg:tabular} (such as the one by~\citet{he2024watermarkinggenerativetabulardata}).

The decoding setting that we consider in this study is where a user (possibly an adversary) is able to ask or `query' the data provider whether the tabular data used by them is watermarked. An adversary can break a large table into multiple chunks of smaller tables that is then queried. Based on the answers (whether or not the table is watermarked), the adversary can then learn the `query function' (which depends on specific parameters used by the data provider) which says whether or not the data is watermarked. Once this query function is learned by the adversary, they can use it as an oracle/labeling mechanism to see if some spoofing technique is able to generate data that the data provider thinks is watermarked. Without access to such a query function, the adversary would potentially need to query the data provider a very large number of times because every query to the data provider potentially takes time/cost which makes running spoofing (\eg brute force) algorithms infeasible.

We demonstrate that the number of queries to the tabular data an adversary would require to learn the watermarking algorithm presented in our work (\cref{alg:tabular}) is $\widetilde{O}(mnb)$, where $b$ is the number of bins used to embed the watermark in \cref{alg:tabular}, and the adversary uses a query table of size $m\times n$. We informally state this result below, and defer the formal statement and its proof to \cref{sec:query_learning}.
\begin{theorem*}[\ref{thm:query-complexity-alg-tab}, Informal]
    A query function for the watermarking scheme in \cref{alg:tabular} can be learned up to error $\epsilon \in (0,1)$ with probability at least $1-\delta \in (0,1)$ by making $O\left(\frac{mnb\log(mnb)\log(1/\epsilon) + \log (1/\delta)}{\epsilon}\right)$ queries to the detector using tables of size $m\times n$, given the number of bins $b$.
\end{theorem*}
In practice, the data provider using \cref{alg:tabular} to watermark their data can artificially limit the total number of queries by an order of magnitudes less than the query complexity of \cref{thm:query-complexity-alg-tab} to avoid divulging the watermarking scheme.
We also observe that the query complexity of learning \cref{alg:tabular} is exactly same as the query complexity of learning the algorithm of \cite{he2024watermarkinggenerativetabulardata}. This is an artifact of how we derive our bounds (via VC dimension). Beyond table like queries, we also study the case when queries are of the form of rows. For the formal theorem statement and an analysis for the row-query case, see \cref{appendix:row_query}.
%
\section{Experiments}
\label{sec:experiment}
\hypertarget{expQ}{In} this section, we empirically evaluate the our watermarking algorithm on synthetic and real datasets to verify our theoretical claims. Formally, we answer the following questions: (\textbf{Q1}) can our watermarking method ensure high data fidelity and downstream utility? (\cref{sec:fidelity}); (\textbf{Q2}) is our watermarking method detectable for datasets of various size and distribution? (\cref{sec:dectection}); (\textbf{Q3}) is our watermarking method robust against oblivious and adversarial attacks? (\cref{sec:robustness}).

\paragraph{Datasets.}
For synthetic data experiments, we generate datasets of various size using the standard Gaussian distribution. For experiments on real datasets, we evaluate on Boston Housing and California Housing Price for regression task, and Adult and Wilt for classification task. For each dataset, we generate corresponding synthetic datasets using Gaussian Copula~\citep{masarotto2012gaussian}. We use Random Forest as the downstream machine learning algorithm for all datasets. For detection, we use a $p$-value threshold of $0.05$.
For additional detail, see \cref{appendix:experiments}.
\begin{figure*}[ht]
    \centering
    \begin{subfigure}[b]{0.32\linewidth}
            \centering
            \includegraphics[width=\textwidth]{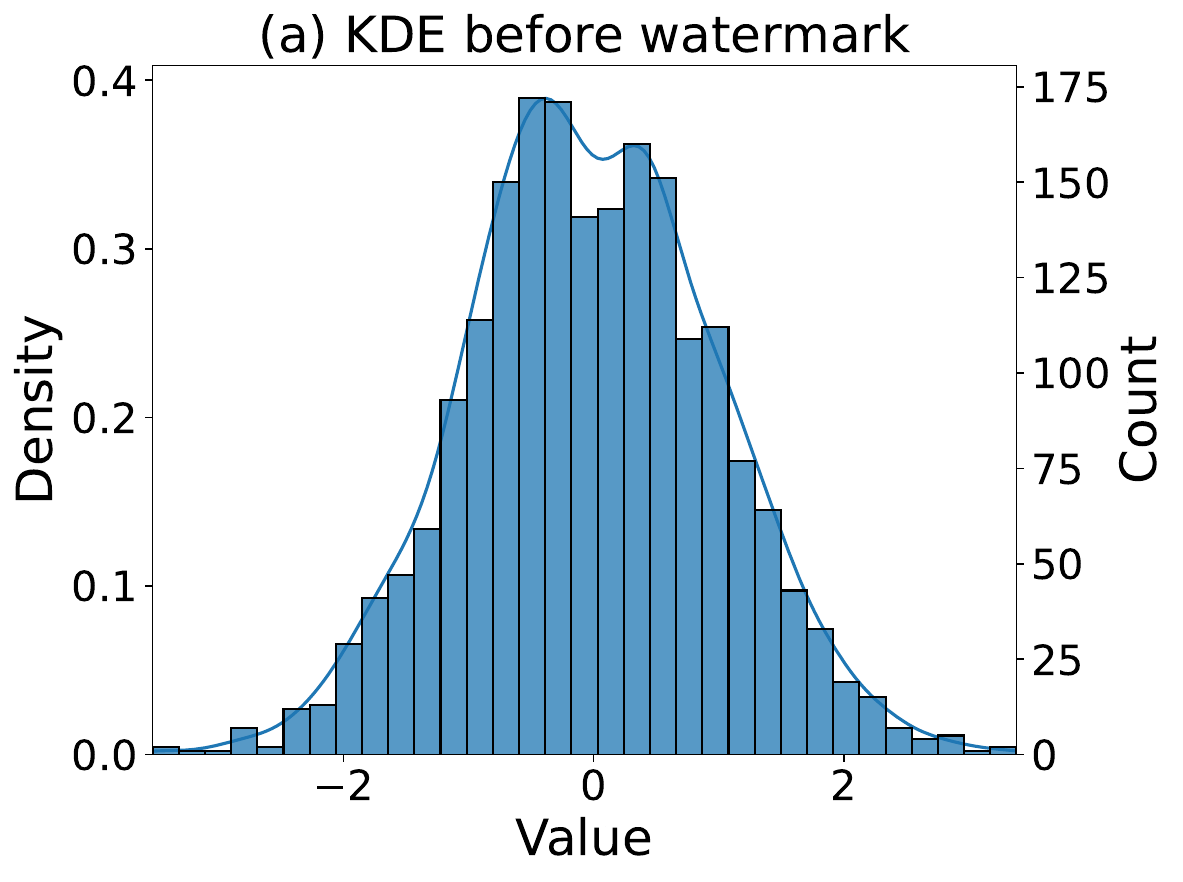}
            \phantomsubcaption
            \label{fig:gaussian-hist}
    \end{subfigure}
    \hfill 
    \begin{subfigure}[b]{0.32\linewidth}
            \centering
            \includegraphics[width=\textwidth]{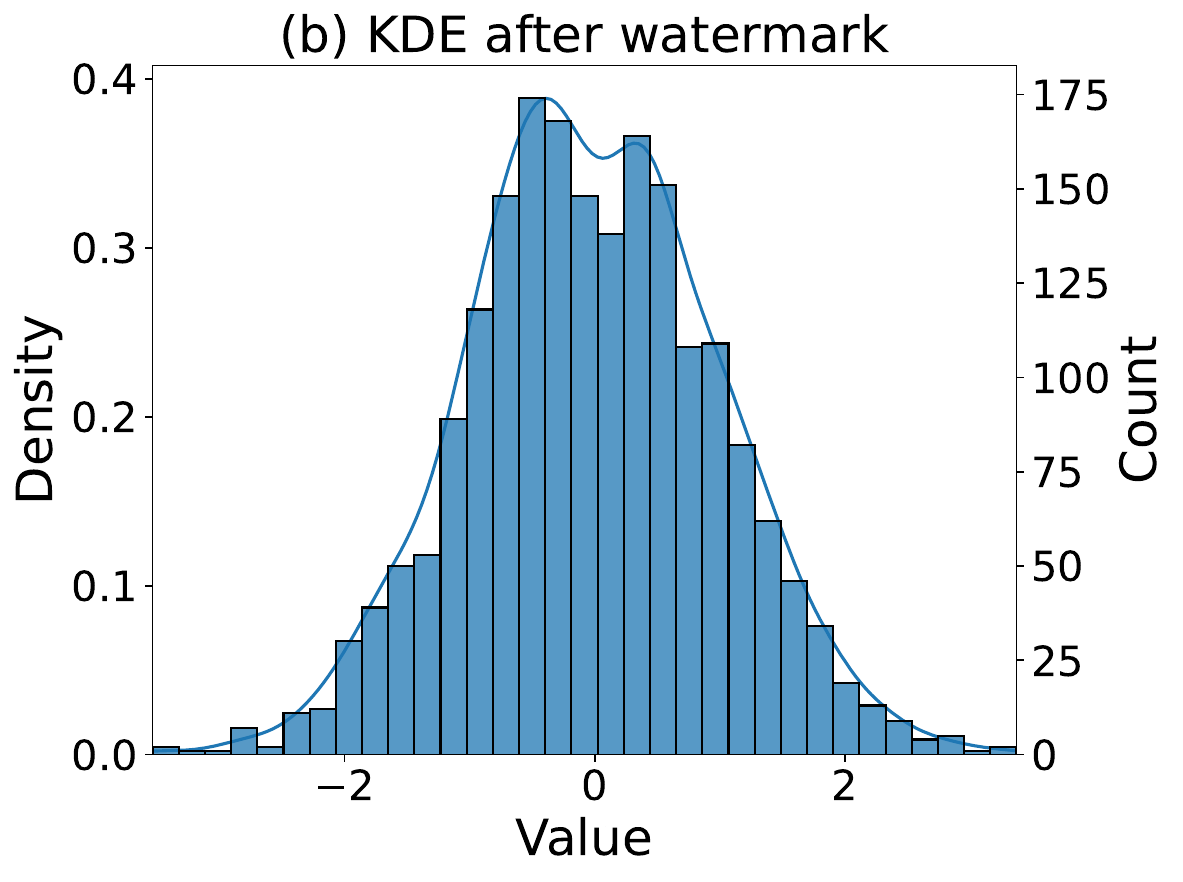}
             \phantomsubcaption
            \label{fig:gaussian-hist-wm}
    \end{subfigure}
    \hfill
    \begin{subfigure}[b]{0.32\linewidth}
            \centering
            \includegraphics[width=\textwidth]{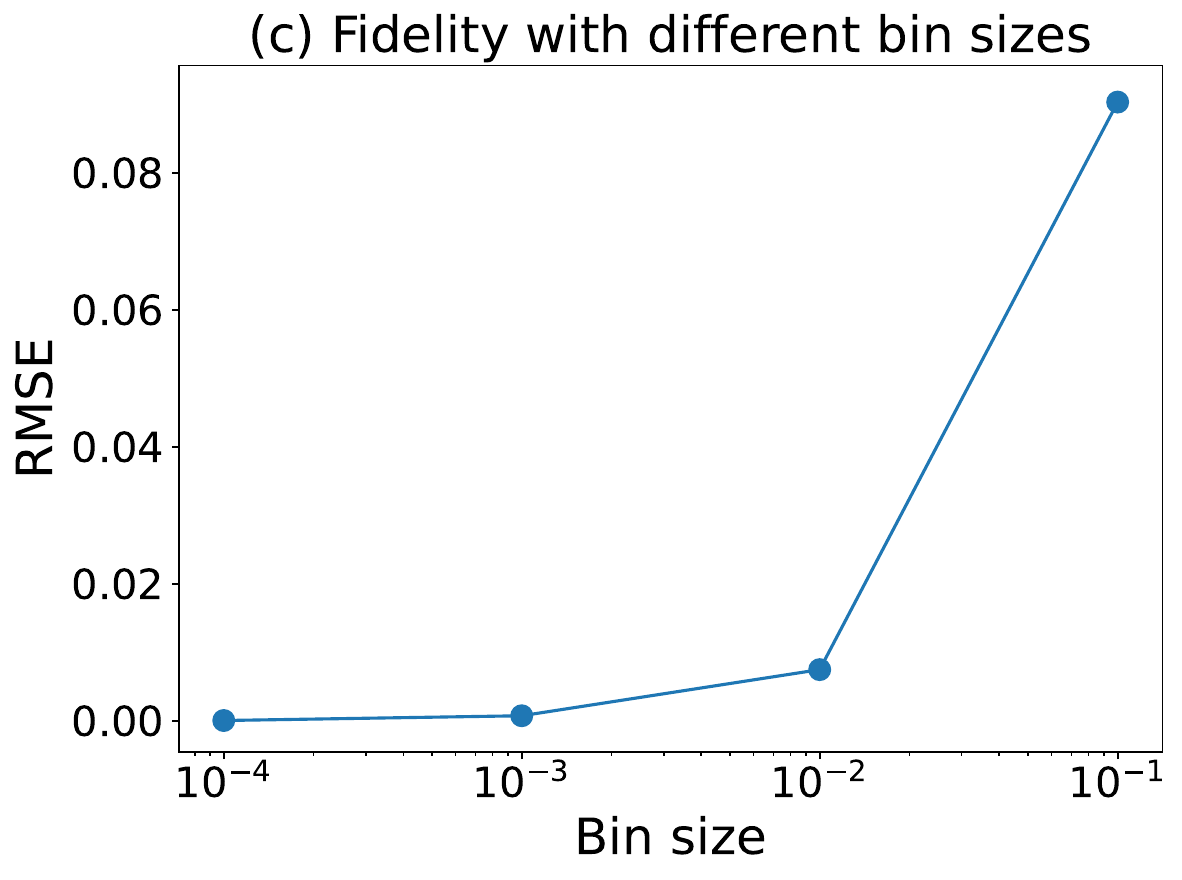}
             \phantomsubcaption
            \label{fig:gaussian-fidelity}
    \end{subfigure}
    
    \begin{subfigure}[b]{0.32\linewidth}
            \centering
            \includegraphics[width=\textwidth]{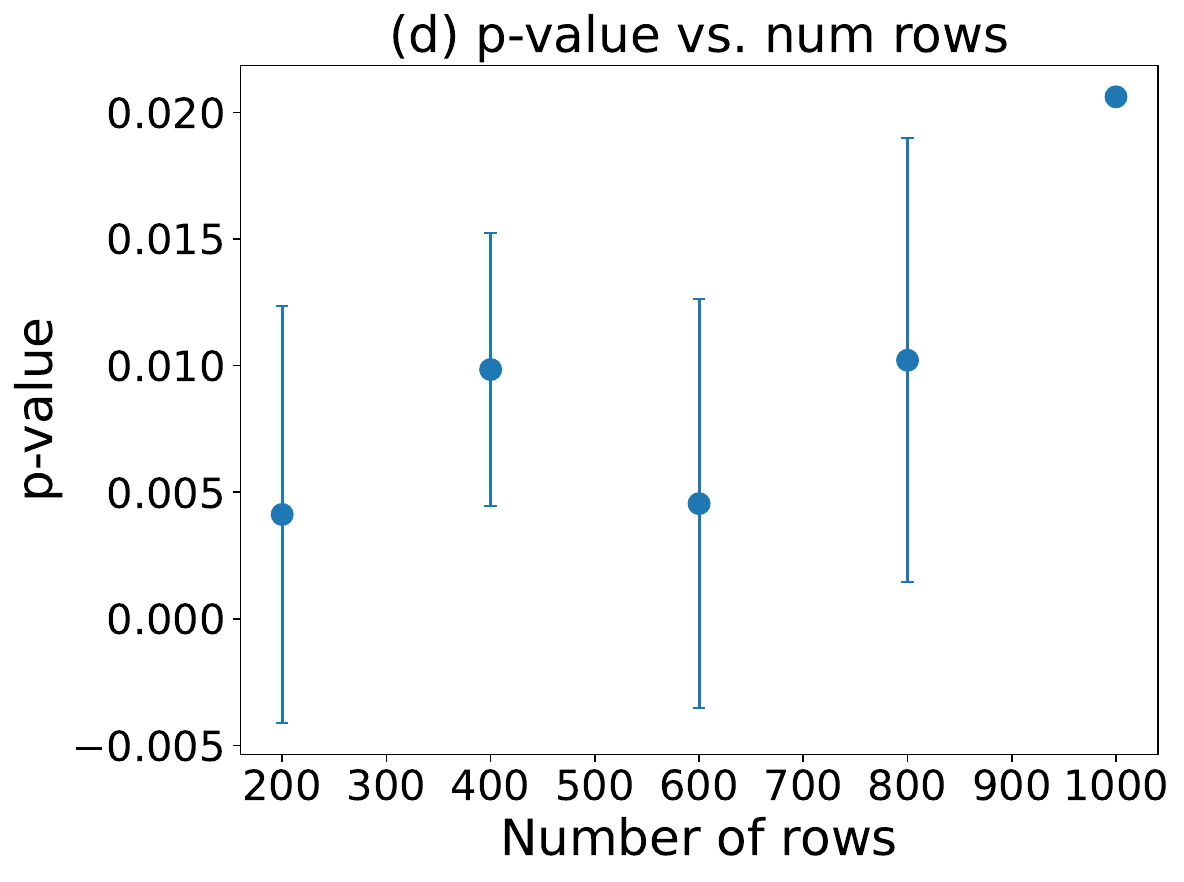}
             \phantomsubcaption
            \label{fig:p-value}
    \end{subfigure}
    \hfill 
    \begin{subfigure}[b]{0.32\linewidth}
            \centering
            \includegraphics[width=\textwidth]{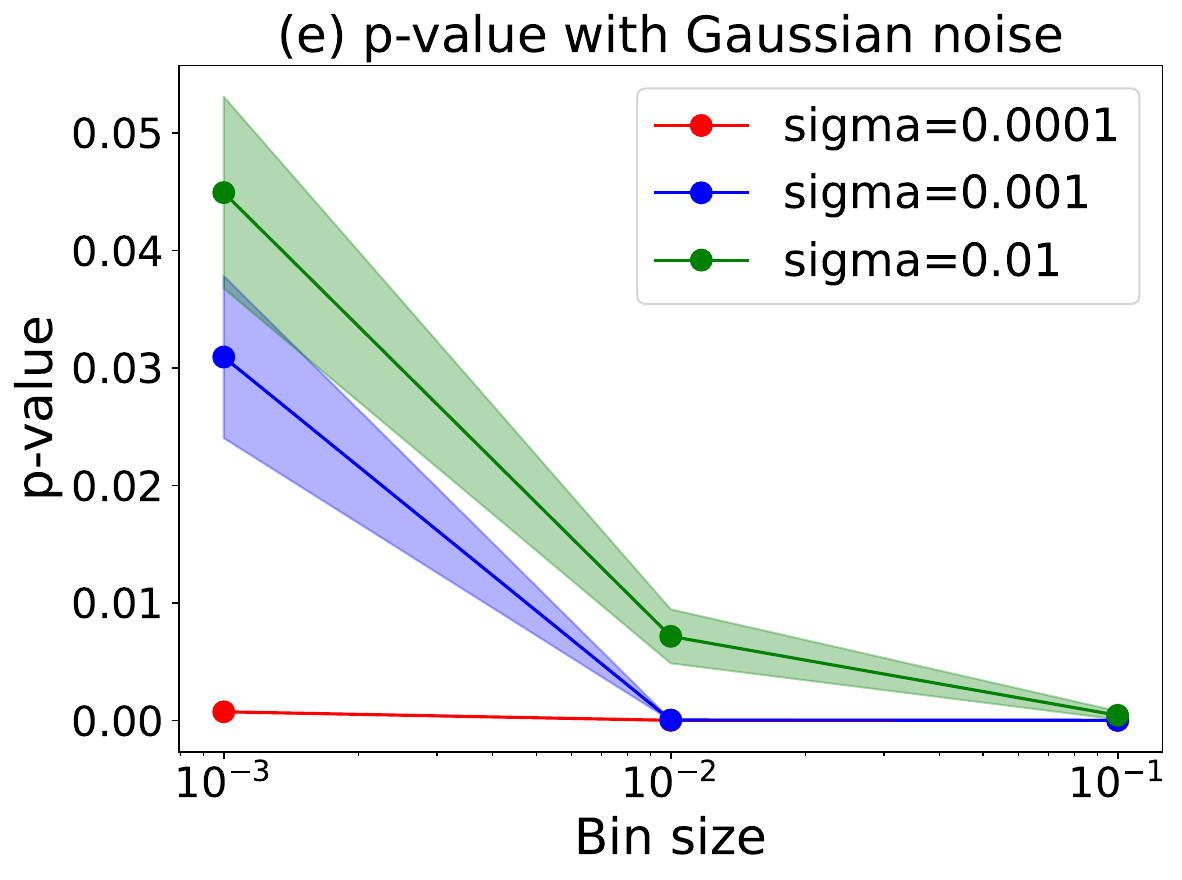}
             \phantomsubcaption
            \label{fig:gaussian-noise}
    \end{subfigure}
    \hfill
    \begin{subfigure}[b]{0.32\linewidth}
            \centering
            \includegraphics[width=\textwidth]{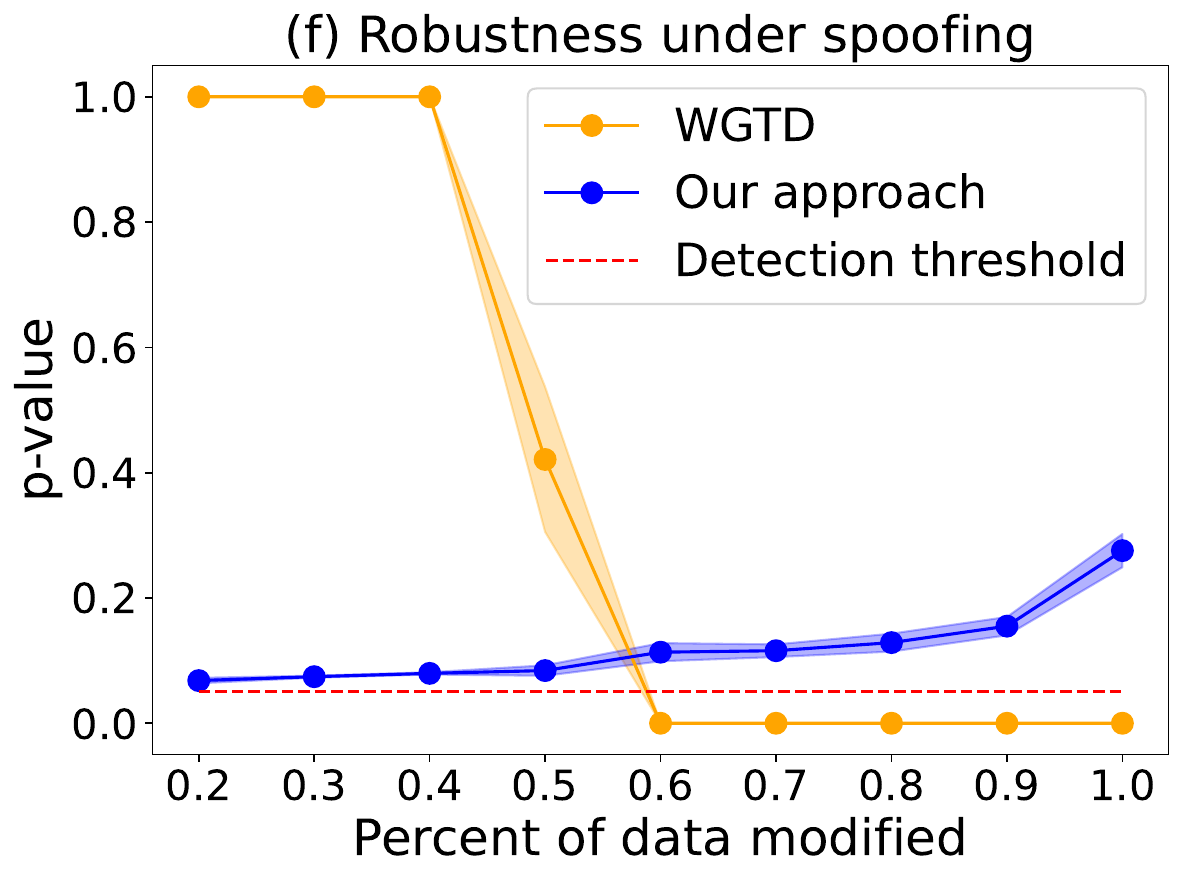}
             \phantomsubcaption
            \label{fig:spoof-data}
    \end{subfigure}
    \caption{
    (a, b) KDE plot of the Gaussian data before and after watermarking respectively. (c) Smaller bin sizes result in higher fidelity (lower RMSE). (d) $p$-value remain below $0.05$ across datasets. (e) Smaller bin sizes are more susceptible to zero-mean Gaussian noise with standard deviation sigma. (f) At $60 \%$ of data processed by \Cref{alg:fractional} (see \Cref{appendix:spoofing}), WGTD can be spoofed while \cref{alg:tabular} remains robust.}
    \label{fig:synthetic-data}
\end{figure*}
\subsection{Fidelity}

\begin{table}[ht]
\centering
\resizebox{\columnwidth}{!}{
\begin{tabular}{llll}
\toprule
           & \textbf{Original} & \textbf{Watermarked} & \textbf{MLE Gap} \\ \midrule
Boston     & 0.330                     & $0.330 \pm 0.000$            & 0.000            \\
California & 84512.544                & $84586.003 \pm 74.140$       & $74.459$         \\
Adult      & 0.741                     & $0.741 \pm 0.000$            & 0.000            \\
Wilt       & 0.496                     & $0.496 \pm 0.000$            & 0.000            \\ 
\bottomrule
\end{tabular}
}
\caption{We report the downstream utility as RMSE for regression task or ROC-AUC score for classification task. The MLE gap is the difference between the downstream performance of watermarked data and real data. For comparison to other baselines see \cref{tab:comparison-fidelity}. While MLE gap seems large for California, the relative gap is of the order of $O(10^{-3})$.}
\label{tab:fidelity}
\end{table}
To verify \hyperlink{expQ}{(Q1)}, we first examine the effect of \cref{alg:tabular} on the data distribution of synthetic datasets using KDE plots (Fig.~\ref{fig:gaussian-hist} and \ref{fig:gaussian-hist-wm}). We observe minimal difference between the data distribution before and after watermarking. 

To compare downstream task utility, we focus on machine learning efficiency gap (MLE Gap)~\citep{xu2019modeling}, which measures the differences in performance of models trained on synthetic data and original data. To compare across different choices of the bin size $b$ in \cref{alg:tabular}, we compute root mean squared error (RMSE) for classification error. Particularly, we vary the bin size $b$ between $10^{-4}$ and $10^{-1}$. We observe that the bin size is directly proportional to RMSE, and so inversely proportional to utility (\cref{fig:gaussian-fidelity}). This is expected as smaller bin size implies that the nearest green interval used for watermarking in \cref{alg-line:embed-watermark} of \cref{alg:tabular} is closer to the original data sample. 

For real-world datasets, we report the RMSE for regression and ROC-AUC score for classification in~\cref{tab:fidelity} over $10$ trials. We observe that the downstream utility of the watermarked dataset does not decrease significantly compared to the original dataset. Thus, we conclude that our \cref{alg:tabular} maintain high data fidelity and downstream utility. 

\subsection{Detectability}
To investigate \hyperlink{expQ}{(Q2)}, we measure the detectability of our watermarking scheme by measuring the $z$-score on the watermarked dataset. We generate synthetic datasets, each containing $50$ columns and vary the number of rows from $200$ to $1000$ to examine the effect of number of samples on the $p$-value. We report the average $p$-value over $5$ trials. We observe that irrespective of the dataset sizes, the $p$-value stays below the $0.05$ threshold (\cref{fig:p-value}), allowing us to always detect the watermark. 

In \cref{tab:p-value-ours}, we report the $p$-value from the detection process described in \cref{sec:dectection} for the watermarked real-world datasets. We observe that the $p$-value for all watermarked datasets are below the $0.05$ threshold (a parameter for the detection algorithm). This result suggests that the watermarking algorithm in \cref{alg:tabular} remains easily detectable across different types of synthetic and real-world datasets.
\begin{table*}[ht]
\centering
\resizebox{\textwidth}{!}{%
\begin{tabular}{llllllllll}
\toprule
\textbf{Dataset} &
  \textbf{Watermarked} &
  \textbf{Row Shuffled} &
  \textbf{Row Del.} &
  \textbf{Col. Del.} &
  \textbf{Importance} &
  \textbf{Gaussian} &
  \textbf{Uniform} &
  \textbf{Alteration} &
  \textbf{Truncation} \\ \midrule
Boston &
  $0.007 \pm 0.000$ &
  $0.012 \pm 0.003$ &
  $0.011 \pm 0.008$ &
  $0.011 \pm 0.007$ &
  $0.007 \pm 0.000$ &
  $0.018 \pm 0.008$ &
  $0.018 \pm 0.008$ &
  $0.018 \pm 0.009$ &
  $0.002 \pm 0.000$ \\ 
California &
  $0.000 \pm 0.000$ &
  $0.000\pm 0.000$ &
  $0.000\pm 0.000$ &
  $0.000\pm 0.000$ &
  $0.000\pm 0.000$ &
  $0.014 \pm 0.008$ &
  $0.018 \pm 0.005$ &
  $0.013 \pm 0.007$ &
  $0.000 \pm 0.000$ \\ 
Adult &
  $0.000 \pm 0.000$ &
  $0.000\pm 0.000$ &
  $0.000\pm 0.000$ &
  $0.000\pm 0.000$ &
  $0.000\pm 0.000$ &
  $0.014 \pm 0.008$ &
  $0.010 \pm 0.009$ &
  $0.018 \pm 0.005$ &
  $0.000 \pm 0.000$ \\ 
Wilt &
  $0.000\pm 0.000$ &
  $0.004 \pm 0.008$ &
  $0.010 \pm 0.009$ &
  $0.087 \pm 0.029$ &
  $0.000\pm 0.000$ &
  $0.014 \pm 0.008$ &
  $0.017 \pm 0.007$ &
  $0.038 \pm 0.040$ &
  $0.021 \pm 0.000$ \\
  \bottomrule
\end{tabular}
}
\caption{$p$-values following several attacks on watermarked datasets.}
\label{tab:p-value-ours}
\end{table*}
\subsection{Robustness}
To examine \hyperlink{expQ}{(Q3)}, we examine the detectability of our watermarking scheme under some oblivious and adversarial attacks. For oblivious attacks, we consider table-level attacks (\eg row shuffling, row deletion, feature selection) and row-level attacks (\eg additive noise, value alteration, truncation). For adversarial attack, we consider a spoofing attack, where the adversary attempts to create a synthetic dataset that can fool the detector in \cref{sec:dectection} without knowledge of the parameters used in \cref{alg:tabular}. 

\paragraph{Oblivious attacks.}
For the synthetic dataset, we first consider the additive noise attack on a dataset of size $2000\times4$. We simulate this attack by adding zero-mean Gaussian noise with the standard deviation varying from $10^{-3}$ to $10^{-1}$ and measure its effects on the $p$-value. We plot this over $5$ trials in \cref{fig:gaussian-noise}. We find that with a similar level of added noise, a smaller bin size results in a higher $p$-value. However, this does not adversely affect detectability.

To validate robustness against feature selection, we create two set of datasets of size $75\times 75$ using \cref{alg:pair} with $p$ as (1) uniform and (2) proportional to feature importance. For each dataset, we train a Random Forest classifier and use the calculated feature importance to drop the least important subsets of columns varying from $20\%$ to $80\%$. We measure the percentage of column pairs preserved over $5$ trials and report the statistics in \cref{fig:robust-pair}. We observe that more $(key,value)$ pairs under the feature importance is preserved compared to uniform pairings. 
\begin{figure*}[ht]
    \centering
    \includegraphics[width=0.9\textwidth]{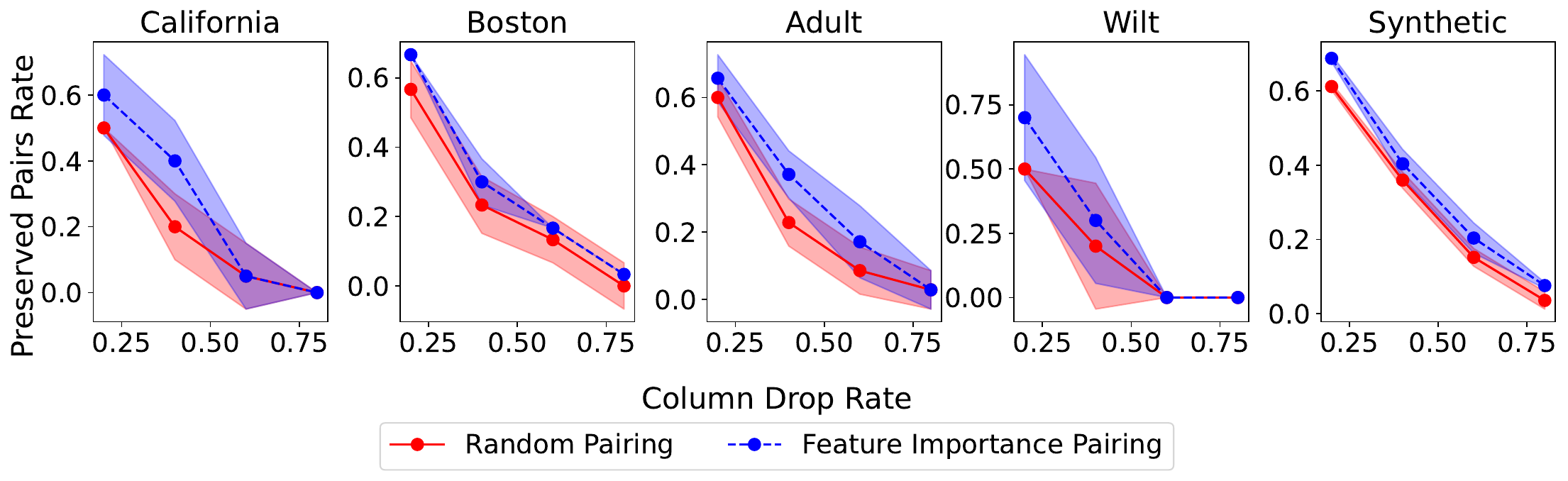}
    \caption{\textbf{Robustness under PAIR.} Each experiment is repeated 5 times, and the mean and standard deviation are reported. Across all datasets, feature importance pairing retains more column pairs than random pairing.}
    \label{fig:robust-pair}
\end{figure*}
For real-world datasets, we employ eight common operations (attacks) on tabular data, divided into two categories: (1) table-level attacks: row shuffling, row deletion, column deletion, feature selection according to feature importance; and (2) cell-level attacks: Gaussian noise addition, uniform noise addition, value alteration, and truncation. We report the mean and standard deviation of the $p$-values from the hypothesis tests in \cref{tab:p-value-ours}, computed over 10 trials. Across all datasets and attacks, the watermark remains detectable (\ie $p$-value remains $\leq 0.05$). We also plot the number of column pairs preserved in each dataset after some columns are dropped due to feature selection under the two pairing routines discussed in \cref{fig:robust-pair}. With increasing number of columns dropped during feature selection, feature importance pairing always retains more column pairs compared to uniformly random pairing. 

\paragraph{Adversarial attack.}
\label{sec:spoofing}
To compare the robustness of our algorithm and the WGTD watermark~\citep{he2024watermarkinggenerativetabulardata} against spoofing attack, we use \cref{alg:fractional} (see \Cref{appendix:spoofing} for details). This spoofing algorithm is novel, and is designed to subtly alter a dataset by modifying the decimal points of its elements so that the dataset appears to be watermarked. For each element, \cref{alg:fractional} identifies the closest fractional part from the watermarked dataset to the current element. This fractional part is then copied over to replace the fractional part of the current element. This ensures that the watermark is embedded in a way that is minimally invasive, preserving the fidelity, detectability, overall structure, and integrity of the original data while introducing a subtle, detectable pattern. The experimental results are in \cref{fig:spoof-data}. 

We observe that one can successfully spoof the WGTD after replacing about $60\%$ of the synthetic dataset. That is, the spoofed dataset generated from \cref{alg:fractional} will be classified as `watermarked' by~\citet{he2024watermarkinggenerativetabulardata}'s detector as the $p$-value drops below $0.05$. On the other hand, our watermark approach (\cref{alg:tabular}) cannot be spoofed by using \cref{alg:fractional} at all. Even after all data points in the dataset have been modified, the $p$-value from our detection hypothesis test still remains above $0.05$. This shows that \cref{alg:tabular} is more robust to simple spoofing attacks simply by leveraging the structure of the tabular dataset. 
\section{Discussion and Future Work}
\label{sec:discussion}
In this work, we propose a simple watermarking scheme for tabular numerical datasets. Our watermarking method partitions the feature space into pairs of $(key, value)$ columns using knowledge of the feature importance in the downstream task. We use the center of the bins from each $key$ column to generate randomized red and green intervals and watermark the $value$ columns by promoting its value to fall in green intervals. Compared to prior work in watermarking generative tabular data, we provide a more rigorous theoretical analysis of different attacks to the watermark. 

Several avenues for future work remain open. This includes adapting our watermarking algorithm for categorical data, or embedding watermarking in the probability spaces pre-generation (similar to methods for LLMs~\citep{venugopal2011watermarking}). This is partly because it is non-trivial to understand how perturbing the space from which samples are drawn by the generative process affects categorical data generation. Cryptographic ideas \citep{kuditipudi2024robust, christ2023undetectable} can be exploited to embed watermarks pre-generation, but understanding the induced distortions remains an open problem. Moreover, the theoretical analysis becomes much more complicated as the generation process is often non-linear. Finally, our analysis for decoding the watermark uses a specific embedding that allows us to apply VC dimensional bounds for learning the query function using neural networks with piecewise linear polynomial functions. It remains open if one can derive better query complexity bounds without resorting only to neural networks.

\begin{acknowledgements}
The authors would like to thank Niccolò Dalmasso and Tyler Chen for valuable discussions during the early stages of this project. 

\paragraph{Disclaimer:}
This paper was prepared for informational purposes by the Artificial Intelligence Research group and the Global Technology Applied Research center of JPMorgan Chase \& Co. and its affiliates (``JP Morgan'') and is not a product of the Research Department of JP Morgan. JP Morgan makes no representation and warranty whatsoever and disclaims all liability, for the completeness, accuracy or reliability of the information contained herein. This document is not intended as investment research or investment advice, or a recommendation, offer or solicitation for the purchase or sale of any security, financial instrument, financial product or service, or to be used in any way for evaluating the merits of participating in any transaction, and shall not constitute a solicitation under any jurisdiction or to any person, if such solicitation under such jurisdiction or to such person would be unlawful.
\end{acknowledgements}

\bibliography{refs}

\newpage

\onecolumn

\title{Adaptive and Robust Watermark for Generative Tabular Data\\(Supplementary Material)}
\maketitle
\equalcontribfootnote{Equal Contribution}
\otherfootnote{Work done while at AI Research, JPMorganChase}

\appendix

\section{Additional Related Work}
\label{appendix:additional-related-work}
\begin{table}[H]
\centering
\resizebox{\textwidth}{!}{
\begin{tabular}{@{}lllll@{}}
\toprule
\multicolumn{2}{c}{} &
  \multicolumn{1}{c}{WGTD~\citep{he2024watermarkinggenerativetabulardata}} &
  \multicolumn{1}{c}{TabularMark~\citep{zheng2024tabularmarkwatermarkingtabulardatasets}} &
  \multicolumn{1}{c}{\textbf{Ours}} \\
\midrule
\multicolumn{2}{l}{Data Type}            & Numerical        & Numerical, Categorical  & Numerical                                          \\
\multicolumn{2}{l}{Fidelity} &
  $\nicefrac{1}{b}$ &
  $\min \left \{ \nicefrac{1}{2b}\left(\log\left(\nicefrac{m}{\delta}\right)-1\right), 1 \right \}$\tablefootnote{\cite{zheng2024tabularmarkwatermarkingtabulardatasets} did not provide theoretical analysis of the fidelity. This bound comes from our analysis with a slightly different union bound event.} &
  $\min \left \{ \nicefrac{1}{2b}\left(\log\left(\nicefrac{mn}{\delta}\right)-1\right), 1 \right \} $ \\
\multicolumn{2}{l}{Detection Test} &
  $\chi^2$-test &
  One Proportion $z$-test &
  \begin{tabular}[c]{@{}l@{}}One Proportion $z$-test\\ with Bonferroni correction\end{tabular} \\
\multicolumn{2}{l}{Robustness} &
  Additive noise &
  \begin{tabular}[c]{@{}l@{}}Additive Noise, Insertion,\\ Deletion, Shuffle,\\ Row Subset, Invertibility\end{tabular} &
  \begin{tabular}[c]{@{}l@{}}Additive Noise, Truncation,\\ Feature Selection\end{tabular} \\
\multicolumn{2}{l}{Spoofing

}             & No               & Yes                     & Yes                                                \\
\multicolumn{2}{l}{Learning} & $\tilde{O}(mnb)$ & $\tilde{O}(mnb)$        & $\tilde{O}(mnb)$                                   \\
\bottomrule
\end{tabular}%
}
\caption{ A comparison of the setting and results from \citep{he2024watermarkinggenerativetabulardata,zheng2024tabularmarkwatermarkingtabulardatasets} and our technique for post-process watermarking generative tabular data. The fidelity measures the $\ell_\infty$-distance between the original and watermarked datasets. For spoofing attack, we list whether the watermark is robust against a simple spoofing procedure (\Cref{sec:spoofing}). The statistical learning upper bounds represent the query complexity required by an adversary to learn the watermark scheme.}
\label{table:comparison}
\end{table}

There exists a long line of research on watermarking relational databases \citep{agrawal2002watermarking, sion2003rights, shehab2007watermarking, lin2021watermark, li2023watermark, kamran2013watermark, hwang2020reversible, hu2019reversible}. \citet{agrawal2002watermarking} is the first work to tackle this problem setting, where the watermark was embedded in the least significant bit of some cells, \ie setting them to be either $0$ or $1$ based on a hash value computed using primary and private keys. Subsequently, \citet{xiao2007robust, hamadou2011weight} proposed an improved watermarking scheme by embedding multiple bits. Another approach is to embed the watermark into the data statistics. Notably, \citet{sion2003rights} partitioned the rows into different subsets and embedded the watermark by modifying the subset-related statistics. This approach is then improved upon by \citet{shehab2007watermarking} to protect against insertion and deletion attacks by optimizing the partitioning algorithm and using hash values based on primary and private keys. To minimize data distortion, the authors modeled the watermark as an optimization problem solved using genetic algorithms and pattern search methods. However, this approach is strictly limited by the requirement for data distribution and the reliance on primary keys for partitioning algorithms. For a comprehensive survey on prior work in watermarking databases, see \citet{kamran2018comprehensive}.

Both the prior works~\citep{he2024watermarkinggenerativetabulardata, zheng2024tabularmarkwatermarkingtabulardatasets} and our proposed approach maintain high data fidelity, measured by comparing the utility of using the watermarked dataset compared to the original unwatermarked dataset on downstream machine learning tasks. While all three approaches use a similar hypothesis testing procedure to detect the watermark, \citet{he2024watermarkinggenerativetabulardata}'s watermark is the easiest to detect. The additional cost of watermark detection in \citet{zheng2024tabularmarkwatermarkingtabulardatasets} and our schemes are due to the selection of the watermarked features. At detection time, the third-party need to identify either the target watermarked feature (for \citet{zheng2024tabularmarkwatermarkingtabulardatasets}'s watermark) or all $(key, value)$ feature pairs (for our approach). Compared to \citet{he2024watermarkinggenerativetabulardata}, our approach is more robust to a simple spoofing attack (\Cref{sec:spoofing}), where an adversary attempts to create harmful content with a target watermark embedded \citep{jovanović2024watermarkstealinglargelanguage}. On the other hand, while \citet{zheng2024tabularmarkwatermarkingtabulardatasets} provided a comprehensive study of potential attacks against their watermark, their approach is less robust to the feature selection attack compared to ours. This difference stems from how the watermarked feature is chosen in their watermark. Since \citet{zheng2024tabularmarkwatermarkingtabulardatasets} mainly watermarks only the target feature (with possible extensions to multiple features), if an oblivious data scientist decides to drop this watermarked feature then the watermark is effectively erased. Meanwhile, we specifically design the $(key, value)$ PAIR routine to preserve the watermark even when a large number of columns is dropped. For security attacks on the watermark, we provide the first upper bound on query complexity required by an adversary to learn the watermark scheme. Under our analysis in \Cref{sec:decoding}, this upper bound is the same for all three tabular watermarking approaches. 

\paragraph{Novelty.} In our watermarking approach, we draw from the vast literature of red-green list watermarking technique for LLMs. However, our application of this idea to tabular data involves judiciously leveraging the tabular structure of the dataset in our PAIR routine (Algorithm~\ref{alg:pair}), which has non-trivial implications for downstream tasks like feature selection (see Figure~\ref{fig:robust-pair} and Appendix~\ref{appendix:feature-selection}). Furthermore, our technical contributions are a comprehensive theoretical study of watermarking continuous tabular data that was previously missing from prior work. From a theoretical standpoint, we improve upon the previous results on several fronts. For detection, \Cref{lem:convergence} in our study holds for every distribution in the non-asymptotic regime, whereas the analogous result in WGTD only holds asymptotically. For robustness, TabularMark only prove results for additive noise attacks for particular distributions or do not account for the perturbed elements to fall back into a green bin (instead of a red bin). We, on the other hand, provide results for arbitrary distributions for additive noise attacks, and additionally analyze truncation and feature selection attacks, which are novel contributions. In addition, we study adversarial attacks to the watermark such as the spoofing attack in \Cref{fig:spoof-data}. Finally, we investigate in detail the decoding of our watermarking algorithm as well as the WGTD algorithm, which includes theoretical (upper) bounds on the query complexity. Our proof technique involves constructing a neural network to decode the watermarks (and obtaining bounds on its VC dimension), which to our knowledge is novel in the watermarking literature. 

Compared to WGTD and TabularMark, our watermark embedding algorithm requires additional time to run the PAIR subroutine (\Cref{alg:pair}). Generating a permutation of the columns takes $O(n)$ time, and creating the pairs take $O(n)$ time in total. The additional memory requirement here is to keep tracks of the pairs, which is stored in a list of size $n$. For watermark detection, we need to check the number of elements in green list intervals for each (key-value) pair, which lead to a multiplicative factor of $O(n)$ in run time more than WGTD and TabularMark. We still need to keep track of the list of (key, value) columns used in watermark embedding at detection time. In simulations, the entire end-to-end process (watermarking embedding, detection, measuring the effect of all oblivious and adversarial attacks) takes less than a minute per run.
\section{Additional experimental details}
\label{appendix:experiments}
For reproducibility, we set a specific random seed to ensure that the data deformation and model training effects are repeatable on a similar hardware and software stack. To eliminate the randomness of the results, the experimental outcomes are averaged over multiple runs from each watermarked training set.

For synthetic data generation using the generative method mentioned in the paper, we follow  the default hyperparameters found in the SDV library ~\citep{patki2016synthetic}. We run the experiments on a machine type of g4dn.4xlarge consisting of 16 CPU, 64GB RAM, and 1 GPU. For fidelity and utility  evaluation, default parameters of Random Forest classifier and regressor with a seed of 42 was used. The synthetic data generation, training and evaluation process typically finishes within 4 hours. Python 3.8 version was used to run the experiments.
\subsection{Datasets.}
\paragraph{Boston Housing Prices.}
This dataset contains information collected by the U.S. Census Service concerning housing in the area of Boston, Massachusetts. It contains 506 entries each of which represent aggregated data about 14 features for homes. Each entry contains the following features such as number of rooms, per capita crime rate by town, pupil teach ratio by town, index accessibility to highways, and value of home of a representative home in the housing record. The attributes are a mixture of continuous and categorical data types. The dataset has a continuous target label indicating the Median value of owner-occupied homes in \$1000's. Due to this, the dataset is used for the regression task to predict the median value of owner-occupied homes in Boston suburbs.
\paragraph{California Housing Prices.}
The dataset was collected from the 1990 U.S. Census and includes various socio-economic and geographical features that are believed to influence housing prices in a given California district\footnote{Data available on the StatLib repository at \url{https://www.dcc.fc.up.pt/~ltorgo/Regression/cal_housing.html}}. It contains 20,640 records and 10 attributes, each of which represents data about homes in the district. Similar to the previous dataset, the attributes are a mixture of continuous and categorical data types. The dataset is used for regression task, with the target variable being the median house value. 
\paragraph{Wilt.}
 This dataset is a public dataset that is part of the UCI Machine Learning Repository\footnote{Data available on the UCI platform at \url{https://archive.ics.uci.edu/dataset/285/wilt}}. It contains data from a remote sensing study that involved detecting diseased trees in satellite imagery. The data set consists of image segments generated by segmenting the pan-sharpened image. The segments contain spectral information from the Quickbird multispectral image bands and texture information from the panchromatic (Pan) image band. This dataset includes 4,889 records and 6 attributes. The attributes are a mixture of numerical and categorical data types. The data set has a binary target label, which indicates whether a tree is wilted or healthy. Therefore, the dataset has a classification task, which is to classify the tree samples as either diseased or healthy.
\paragraph{Adult.} This is a public dataset that is part of the UCI Machine Learning Repository\footnote{Data available on the UCI platform at \url{https://archive.ics.uci.edu/dataset/2/adult}}. It contains census data from 1994 with demographic information for a subsection of the US population. This dataset includes 48842 records and 14 attributes. The attributes are a mixture of numerical and categorical data types. The label for each record is a binary variable indicating whether an individual's annual income exceeds \$50,000. Thus, the learning task is that of binary classification.  

\subsection{Additional Fidelity Experiment}
\begin{table}[ht]
\centering
\resizebox{\textwidth}{!}{%
\begin{tabular}{@{}lcccccc@{}}
\toprule
\textbf{Dataset}  & \textbf{California} & \textbf{Boston} & \textbf{Adult} & \textbf{Wilt} & \textbf{Gaussian} & \textbf{Uniform} \\ \midrule 
Distance & $7.456$e-05 $\pm 1.588$e-07           & $9.235$e-05 $\pm 1.593$e-05       & $6.900$e-05 $\pm 2.356$e-07      & $2.511$e-04 $\pm 2.790$e-06     & $9.784$e-05 $\pm 9.729$e-07         & $4.768$e-05 $\pm 8.984$e-07       \\
\bottomrule
\end{tabular}
}
\caption{The Wasserstein-1 distance between the original dataset and the watermarked dataset. All experiments are repeated 5 times and the mean and standard deviation are reported.}
\label{tab:ell-1-distance}
\end{table}
First, we want to clarify that the fidelity results for California Housing dataset seems abnormal due to the distribution of the target variable 'median\_house\_value'. This label ranges from $\$14999-\$500000$, and follows a right-skewed distribution with a notable spike at the far right end of the distribution. This peak is a common artifact in housing datasets where values are capped or clipped at a maximum threshold. Comparatively, the variance and MLE Gap in the fidelity experiment is of order of $10^{-4}$ w.r.t. the range. Hence, the result for this dataset is a direct result of intrinsic dataset characteristic.

We further examine the effect of our watermarking algorithm~\ref{alg:tabular} on fidelity for both real-world and synthetic datasets. We generate two types of synthetic dataset: Gaussian dataset of size $10000 \times 10$ with samples drawn i.i.d. from the standard Gaussian distribution, and Uniform dataset of the same size with samples drawn i.i.d. from the $\unif[0,1]$ distribution. All datasets are watermarked using \Cref{alg:tabular} with $1000$ bins. 

In \Cref{tab:ell-1-distance}, we first flatten both the original dataset and the watermarked dataset, then calculate the entry-wise $\ell_1$ distance between these flattened datasets. The experiments are repeated $5$ times, and we report the mean and standard deviation of the results. Across all datasets, we observe small differences between the original and watermarked datasets, showing that our watermarking approach maintain high fidelity.   
Additionally, we evaluate the fidelity of the watermarking procedure under different tabular generators: TVAE and CTGAN~\citep{xu2019modeling} (\Cref{tab:fidelity-tvae-ctgan}) and downstream learners: XGBoost and Linear/Logistic Regression (\Cref{tab:fidelity-learners}). We observe that the performance of our watermark in these experiments stays consistent with our prior results in \Cref{sec:experiment}.
\begin{table}[ht]
\centering
\begin{tabular}{@{}lllllllll@{}}
\toprule
      & Boston & MLE Gap & California & MLE Gap & Adult & MLE Gap & Wilt  & MLE Gap \\ \midrule
TVAE  & 0.086  & 0.000   & 91960.528  & 0.000   & 0.793 & 0.000   & 0.895 & 0.000   \\
CTGAN & 0.438  & 0.000   & 81883.920  & 0.000   & 0.899 & 0.000   & 0.617 & 0.002   \\ \bottomrule
\end{tabular}
\caption{We report the downstream utility as either
RMSE for regression task or ROC-AUC score for classification task. The MLE gap is the difference between
the downstream performance of watermarked data and
real data.}
\label{tab:fidelity-tvae-ctgan}
\end{table}

\begin{table}[ht]
\resizebox{\textwidth}{!}{%
\begin{tabular}{@{}lllllllllllllllll@{}}
\toprule
      & \multicolumn{4}{l}{Boston}    & \multicolumn{4}{l}{California}        & \multicolumn{4}{l}{Adult}     & \multicolumn{4}{l}{Wilt}      \\ \midrule
 &
  XGBoost &
  MLE Gap &
  Linear &
  MLE Gap &
  XGBoost &
  MLE Gap &
  Linear &
  MLE Gap &
  XGBoost &
  MLE Gap &
  Linear &
  MLE Gap &
  XGBoost &
  MLE Gap &
  Linear &
  MLE Gap \\
GaussianCopula &
  0.359 &
  0.000 &
  0.349 &
  0.000 &
  85382.301 &
  158.94 &
  80135.876 &
  0.000 &
  0.469 &
  0.000 &
  0.436 &
  0.000 &
  0.032 &
  0.000 &
  0.031 &
  0.000 \\
TVAE  & 0.098 & 0.000 & 0.098 & 0.000 & 81801.875 & 0.000 & 86054.062 & 0.000 & 0.310 & 0.000 & 0.351 & 0.000 & 0.096 & 0.000 & 0.091 & 0.000 \\
CTGAN & 0.466 & 0.000 & 0.434 & 0.000 & 85961.290 & 0.000 & 86822.823 & 0.000 & 0.361 & 0.000 & 0.393 & 0.000 & 0.325 & 0.000 & 0.308 & 0.000 \\ \bottomrule
\end{tabular}
}
\caption{We report the downstream utility as either
RMSE for regression task or ROC-AUC score for classification task. The MLE gap is the difference between
the downstream performance of watermarked data and
real data.}
\label{tab:fidelity-learners}
\end{table}
\subsection{Additional Robustness Experiment}
In this experiment, we investigate the $p$-value on the watermarked dataset after it has been modified by some adversary. For the synthetic dataset, we generate a dataset of size $10000 \times 21$ for a multi-class classification problem using the \textsc{scikit-learn} package. Specifically, the samples in the features columns are drawn i.i.d. from a standard Gaussian distribution, and among them, there are $4$ informative features, $4$ redundant features (\ie random linear combinations of the informative features), $4$ duplicated features (\ie drawn randomly from the informative and redundant features), and the rest are filled with random noise. The label column contains $5$ different classes. Half of the columns are randomly chosen as the $key$ columns, and the other half are $value$ columns. We choose to perform experiment on this synthetic dataset to have more control over the features parameters in the datasets.  

In \Cref{tab:col-drop}, we report the $p$-value of the watermarked dataset after some columns have been dropped randomly. The proportion of column drops varies between $0.2$ and $0.8$. We repeat this experiment $5$ times and report both the mean and average values. We observe that the $p$-value for all configurations are below the $0.05$ threshold, and we can conclude that the presence of the watermark can always be detected after some columns have been dropped.

\begin{table}[ht]
\centering
\resizebox{\textwidth}{!}{%
\begin{tabular}{@{}lccccccc@{}}
\toprule
\textbf{Column Drop Proportion} &
  $0.2$ &
  $0.3$ &
  $0.4$ &
  $0.5$ &
  $0.6$ &
  $0.7$ &
  $0.8$ \\ \midrule
$p$-value &
  $0.017 \pm 0.007$ &
  $0.014 \pm 0.008$ &
  $0.013 \pm 0.009$ &
  $0.012 \pm 0.008$ &
  $0.010 \pm 0.005$ &
  $0.042 \pm 0.018$ &
  $0.030 \pm 0.018$ \\
  \bottomrule
\end{tabular}
}
\caption{The calculated $p$-value after some columns have been dropped from the watermarked dataset. All experiments are repeated $5$ times and the mean and standard deviation are reported.}
\label{tab:col-drop}
\end{table}
In the next experiment, we examine the $p$-value after the watermarked dataset has been `attacked' by an adversary adding some zero-mean noise. Specifically, we consider 4 types of noise distributions: (i) Uniform: noise drawn i.i.d. from $\unif[0,1]$, (ii) Gaussian: noise drawn i.i.d. from $\cN(0,1)$, (iii) Laplace: noise drawn i.i.d. from a zero-mean Laplace distribution with scale $1$, and (iv) Exponential: noise drawn i.i.d from a zero-mean Exponential distribution with rate $\lambda = 1$. We repeat this experiment $5$ times and report the mean and standard deviation of the calculated $p$-value. We observe that across all attacks, the $p$-value is below the $0.05$ threshold and we can successfully detect the watermark even after the watermarked dataset has been tampered with using additive noise.

\begin{table}[ht]
\centering
\begin{tabular}{@{}ccccc@{}}
\toprule
\textbf{Noise Type} & \textbf{Uniform}            & \textbf{Gaussian}          & \textbf{Laplace}           & \textbf{Exponential}       \\ \midrule
$p$-value    & $0.013 \pm  0.007$ & $0.015 \pm 0.007$ & $0.015 \pm 0.007$ & $0.008 \pm 0.006$ \\
\bottomrule
\end{tabular}
\caption{The calculated $p$-value after additive noise attack with different noise distributions. All experiments are repeated $5$ times and the mean and standard deviation are reported.}
\label{tab:p-value-noise}
\end{table}

\subsection{Comparison to prior work}
\label{sec:experiment-comparison}
We investigate the performance of our watermarking technique on four different real-world datasets and compare it with prior watermarking algorithms as benchmark.

\paragraph{Experimental Setup.} We evaluate our watermark and the baseline approaches on four real-world datasets: \textbf{Boston Housing} and \textbf{California Housing Price} for regression task, as well as \textbf{Adult} and \textbf{Wilt} for classification task. For each dataset, we generate corresponding synthetic datasets using Gaussian Copula \citep{masarotto2012gaussian}. We use Random Forest as the downstream machine learning algorithm for all datasets. 

\paragraph{Baseline approaches.} We compare our performance to three other tabular watermarking techniques: WGTD \citep{he2024watermarkinggenerativetabulardata}, TabularMark \citep{zheng2024tabularmarkwatermarkingtabulardatasets}, and RINTAW \citep{fangrintaw}. As the source codes are unavailable, we implement these watermarking techniques using the authors' specifications. That is, for WGTD implementation, we set the number of `green list' intervals to be $m = 500$, for a total of $1000$ intervals between $0$ and $1$. For TabularMark implementation, we always embed the watermark into $10\%$ of the label column, with number of unit domains $k = 500$, perturbation range $p = 25$. For classification task, we select a half of the possible categories to be `green' instead of using the aforementioned unit domains and perturbation range. Finally, for RINTAW implementation, we use the default masking ratio of $0$.

\paragraph{Experimental results.} First, we study the fidelity of our watermarking approach on these datasets. Particularly, we are interested in measuring both downstream performance metric (RMSE for regression task and ROC-AUC score for classification task) and the Machine Learning Efficiency (MLE), \ie the gap between the performance of watermarked data and the real test data. Each experiment is repeated $10$ times and the mean and standard deviation for the target performance metric are reported. The result of this experiment is summarized in \Cref{tab:comparison-fidelity}. Overall, RINTAW achieves the best downstream performance among all watermarking methods. However, our approach achieves the smallest MLE Gap in three datasets (Boston Housing, Adult, and Wilt), and the second smallest MLE gap in the remaining dataset (California Housing). That is, for both regression and classification, our watermark approach leads to minimal distortion compared to the real, unwatermarked dataset. 

\begin{table}[ht]
\resizebox{\textwidth}{!}{%
\begin{tabular}{@{}llrlrlrlr@{}}
\toprule
              & \multicolumn{2}{c}{Boston}       & \multicolumn{2}{c}{California} & \multicolumn{2}{c}{Adult}        & \multicolumn{2}{c}{Wilt}         \\ \midrule
            &\multicolumn{1}{c}{RMSE} &
  \multicolumn{1}{c}{MLE Gap} &
  \multicolumn{1}{c}{RMSE} &
  \multicolumn{1}{c}{MLE Gap} &
  \multicolumn{1}{c}{ROC-AUC} &
  \multicolumn{1}{c}{MLE Gap} &
  \multicolumn{1}{c}{ROC-AUC} &
  \multicolumn{1}{c}{MLE Gap} \\ \midrule
\textbf{Ours} & $0.330 \pm 0.000$ & $\mathbf{0.000}$ & $84586.003 \pm 74.140$  & $74.459$              & $0.741 \pm 0.000$ & $\mathbf{0.000}$              & $0.496 \pm 0.000$ & $\mathbf{0.000}$              \\
WGTD        & $0.334 \pm 0.005$ & $0.004$  & $84832.096 \pm 220.436$  & $319.557$ & $0.742 \pm 0.004$ & $0.004$ & $0.496 \pm 0.000$ & $\mathbf{0.000}$  \\
TabularMark & $0.987 \pm 0.048$ & $0.656$  & $\mathbf{84521.708 \pm 59.071}$ & $\mathbf{9.164}$    & $0.742 \pm 0.004$ & $0.001$  & $0.534 \pm 0.157$   & $0.037$  \\
RINTAW      & $\mathbf{0.328 \pm 0.002}$  & $0.001$ & $85047.119 \pm 601.940$ & $534.575$  & $\mathbf{0.738 \pm 0.004}$ & $0.003$ & $\mathbf{0.493 \pm 0.002}$  & $0.003$ \\ \bottomrule
\end{tabular}%
}
\caption{We report the RMSE for regression task and ROC-AUC score for classification task. Each experiment is repeated 10 times and the mean and standard deviation are reported. Furthermore, we also report the MLE gap, which is the difference between the downstream performance of watermarked data and real test data. The best results in each column is highlighted in bold.}
\label{tab:comparison-fidelity}
\end{table}

Second, we investigate the detectability and robustness of our watermark under different threat models. Particularly, we employ eight common operations (attacks) on tabular data, divided into two categories: (1) table level attacks: row shuffling, row deletion, column deletion, feature selection according to feature importance; and (2) cell level attacks: Gaussian noise addition, uniform noise addition, value alteration, and truncation. For row deletion, we randomly remove $20\%$ of rows. For column deletion, we randomly remove 2 feature columns. For feature selection according to feature importance, we keep the top $40\%$ of columns with the highest feature importance. For Gaussian noise addition, we inject i.i.d. Gaussian noise with mean zero and variance one into numeric columns. For uniform noise addition, we inject noise drawn i.i.d. from $\unif[-0.1, 0.1]$ into numeric columns. For value alteration, we modify the numeric columns by multiplying them with a random factor drawn uniformly at random from $[0.9, 1.1]$, \ie increase or decrease each numeric cell by at most $10\%$ of their original value. For truncation, we truncate all numeric float values to only 2 decimal places. For each attack, we report the mean $p$-value from the detection procedure, where the dataset is detected as watermarked if $p$-value $< 0.05$. The result of this experiment is summarized in \Cref{tab:comparison-detectability}. Among all considered watermarking techniques, our method is the most consistent at detecting the embedded watermark under all threat models. On the other hand, WTGD is generally robust to table level attacks (row deletion, column deletion, feature importance column deletion), but is more susceptible to additive noise attacks. TabularMark is the least robust method, as it often fails to detect watermark under most cell level attacks (additive noise, value alteration, and truncation). Finally, RINTAW is very robust to row deletion and row shuffling, but may fail to detect the watermark under column deletion and cell level attacks. Overall, our watermarking technique achieves both minimal distortion (small MLE gap in \Cref{tab:comparison-fidelity}) and high fidelity/robustness (consistently low $p$-value in \Cref{tab:comparison-detectability}). 

\begin{table}[ht]
\resizebox{\textwidth}{!}{%
\begin{tabular}{llrrrrrrrrr}
\toprule
 &  & Watermarked & Row Shuffled & Row Del. & Col. Del. & Importance & Gaussian & Uniform & Alteration & Truncation \\
Dataset & Method &  &  &  &  &  &  &  &  &  \\
\midrule
\multirow[t]{4}{*}{Boston} & \textbf{Ours} & $0.007 \pm 0.000$ & $0.012 \pm 0.003$ & $0.011 \pm 0.008$ & $0.011 \pm 0.007$ & $0.007 \pm 0.000$ & $0.018 \pm 0.008$ & $0.018 \pm 0.008$ & $0.018 \pm 0.009$ & $0.002 \pm 0.000$ \\
 & WTGD & $0.000 \pm 0.000$ & $0.000 \pm 0.000$ & $0.000 \pm 0.000$ & $0.000 \pm 0.000$ & $1.000 \pm 0.490$ & $0.462 \pm 0.011$ & $0.418 \pm 0.015$ & $0.000 \pm 0.000$ & $0.000 \pm 0.000$ \\
 & TabularMark & $0.000 \pm 0.000$ & $0.000 \pm 0.000$ & $0.000 \pm 0.000$ & $1.000 \pm 0.400$ & $1.000 \pm 0.000$ & $1.000 \pm 0.000$ & $1.000 \pm 0.000$ & $1.000 \pm 0.000$ & $1.000 \pm 0.000$ \\
 & RINTAW & $0.000 \pm 0.000$ & $0.000 \pm 0.000$ & $0.000 \pm 0.000$ & $0.000 \pm 0.000$ & $0.529 \pm 0.212$ & $0.594 \pm 0.137$ & $0.581 \pm 0.381$ & $0.386 \pm 0.315$ & $0.000 \pm 0.000$ \\
 \cline{1-11}
\multirow[t]{4}{*}{California} & \textbf{Ours} & $0.000 \pm 0.000$ & $0.000\pm 0.000$ & $0.000\pm 0.000$ & $0.000\pm 0.000$ & $0.000\pm 0.000$ & $0.014 \pm 0.008$ & $0.018 \pm 0.005$ & $0.013 \pm 0.007$ & $0.000 \pm 0.000$ \\
 & WTGD & $0.000\pm 0.000$ & $0.000\pm 0.000$ & $0.000\pm 0.000$ & $0.000\pm 0.000$ & $0.000\pm 0.000$ & $0.482 \pm 0.006$ & $0.273 \pm 0.161$ & $0.055 \pm 0.100$ & $0.000 \pm 0.000$ \\
 & TabularMark & $0.000\pm 0.000$ & $0.000\pm 0.000$ & $0.000\pm 0.000$ & $1.000\pm 0.000$ & $1.000\pm 0.000$ & $1.000\pm 0.000$ & $1.000\pm 0.000$ & $1.000\pm 0.000$ & $0.000\pm 0.000$ \\
 & RINTAW & $0.000\pm 0.000$ & $0.000\pm 0.000$ & $0.000\pm 0.000$ & $0.000\pm 0.000$ & $0.022 \pm 0.027$ & $0.520 \pm 0.325$ & $0.820 \pm 0.105$ & $0.442 \pm 0.046$ & $0.000\pm 0.000$ \\
\cline{1-11}
\multirow[t]{4}{*}{Adult} & \textbf{Ours} & $0.000 \pm 0.000$ & $0.000\pm 0.000$ & $0.000\pm 0.000$ & $0.000\pm 0.000$ & $0.000\pm 0.000$ & $0.014 \pm 0.008$ & $0.010 \pm 0.009$ & $0.018 \pm 0.005$ & $0.000 \pm 0.000$ \\
 & WTGD & $0.000\pm 0.000$ & $0.000\pm 0.000$ & $0.000\pm 0.000$ & $0.000\pm 0.000$ & $0.000\pm 0.000$ & $0.446 \pm 0.012$ & $0.376 \pm 0.186$ & $0.000\pm 0.000$ & $0.000\pm 0.000$ \\
 & TabularMark & $0.039 \pm 0.039$ & $0.358 \pm 0.277$ & $1.000 \pm 0.000$ & $1.000 \pm 0.000$ & $1.000 \pm 0.000$ & $1.000\pm 0.000$ & $1.000\pm 0.000$ & $1.000\pm 0.000$ & $0.398 \pm 0.390$ \\
 & RINTAW & $0.000\pm 0.000$ & $0.000\pm 0.000$ & $0.000\pm 0.000$ & $0.000\pm 0.000$ & $0.358 \pm 0.150$ & $0.561 \pm 0.313$ & $0.530 \pm 0.246$ & $0.667 \pm 0.149$ & $0.000 \pm 0.000$ \\

\cline{1-11}
\multirow[t]{4}{*}{Wilt} & \textbf{Ours} & $0.000\pm 0.000$ & $0.004 \pm 0.008$ & $0.010 \pm 0.009$ & $0.087 \pm 0.029$ & $0.000\pm 0.000$ & $0.014 \pm 0.008$ & $0.017 \pm 0.007$ & $0.038 \pm 0.040$ & $0.021 \pm 0.000$s \\
 & WTGD & $0.034 \pm 0.016$ & $0.338 \pm 0.167$ & $0.345 \pm 0.168$ & $0.381 \pm 0.147$ & $1.000\pm 0.000$ & $0.483 \pm 0.014$ & $0.481 \pm 0.010$ & $0.486 \pm 0.006$ & $0.090 \pm 0.124$ \\
 & TabularMark & $0.000\pm 0.000$ & $0.000\pm 0.000$ & $0.000\pm 0.000$ & $1.000\pm 0.000$ & $1.000\pm 0.000$ & $1.000\pm 0.000$ & $1.000\pm 0.000$ & $1.000\pm 0.000$ & $1.000\pm 0.000$ \\
 & RINTAW & $0.000\pm 0.000$ & $0.000\pm 0.000$ & $0.000\pm 0.000$ & $0.000\pm 0.000$ & $0.626 \pm 0.130$ & $0.584 \pm 0.340$ & $0.363 \pm 0.342$ & $0.440 \pm 0.280$ & $0.700 \pm 0.109$ \\
\cline{1-11}
\bottomrule
\end{tabular}
}
\caption{We report the $p$-value of each detection algorithm for different threat models and datasets. Each experiment is repeated 10 times and the mean and standard deviation are reported. Except for Wilt dataset and column deletion attack, our watermark can always be reliably detected for the remaining datasets and threat models, \ie $p < 0.05$. On the other hand, WTGD often fails to detect the watermark under Gaussian noise addition and uniform noise addition. Furthermore, TabularMark fails to detect the watermark in most cell level attacks (additive noise, value alteration, and truncation). Finally, RINTAW may fail to detect the watermark under feature importance column drop and cell level attacks.}
\label{tab:comparison-detectability}
\end{table}

\subsection{Implementation Details}
In this subsection, we provide the pseudo-code of our core implementations: PAIR subroutine (\Cref{alg:pair}), Pairwise Tabular Watermarking (\Cref{alg:tabular}), and the detection process (\Cref{sec:dectection}). 

\lstdefinestyle{algo}{
  language=Python,
  basicstyle=\ttfamily\small,
  columns=fullflexible,
  keepspaces=true,
  showstringspaces=false,
  frame=single,
  framerule=0.8pt,
  rulecolor=\color{black},
  xleftmargin=0pt,
  xrightmargin=0pt,
  aboveskip=8pt,
  belowskip=8pt
}

\begin{lstlisting}[style=algo,title={Helper function: Create Green List},label={lst:create_green_list}]
Algorithm CREATE_GREEN_LIST(value, bins, green_list_percent):
    decimal := fractional_part(value)
    b := digitize(decimal, bins)                # bin index for decimal part
    h := hash(b)
    rng_h := RNG(seed=h)                        # deterministic per seed value

    choices := [1, 2, ..., len(bins)-1]         
    k := floor(len(bins) * green_list_percent)
    green := rng_h.choice(choices, size=k, replace=False)

    return green
\end{lstlisting}

\begin{lstlisting}[style=algo,title={Algorithm 1: PAIR subroutine},label={lst:pair_stage}]
Algorithm PAIR_COLUMNS(X, labels=None, min_decimals=2, pair_percent=1):
    # 1) Build sampling distribution over columns
    if labels is provided:
        importance := RandomForestFeatureImportance(X, labels)
    else:
        importance := uniform_over_columns(X)
    importance := smooth_zeros_and_normalize(importance)

    # 2) Sample a column order and keep valid float columns
    order := sample_without_replacement(columns(X), p=importance)
    valid := valid_float_columns(X, min_decimals)
    cols := filter(order, valid)

    # 3) Pair neighboring columns in sampled order
    seed_cols := cols[0::2]
    mark_cols := cols[1::2]
    if len(seed_cols) > len(mark_cols):
        seed_cols := seed_cols[:-1]

    n_pairs := min(ceil(len(seed_cols) * pair_percent), len(seed_cols))
    pairs := []
    for i in 0..n_pairs-1:
        pairs.append((seed_cols[i], mark_cols[i]))

    return pairs
\end{lstlisting}

\begin{lstlisting}[style=algo,title={Algorithm 2: Watermark},label={lst:watermark_stage}]
Algorithm WATERMARK(X, pairs, bins, green_list_percent, corresponding_decimals):
    function WM_ONE(seed_val, mark_val):
        green := CREATE_GREEN_LIST(seed_val, bins, green_list_percent)

        d, nd := decimal_part_and_precision(mark_val)
        nd := max(nd, corresponding_decimals)
        b := digitize(d, bins)

        if b in green:
            return mark_val

        g := nearest(green, b)
        d_new := uniform(bins[g-1], bins[g])    # sample in nearest green bin
        d_new := truncate_to_decimals(d_new, nd)

        return round(round(mark_val - d, nd) + d_new, nd)

    W := copy(X)
    pair_map := {}

    for (s, t) in pairs:
        for r in rows(X):
            W[r, t] := WM_ONE(X[r, s], X[r, t])
        pair_map[s] := t

    return W, {"pairs": pair_map}
\end{lstlisting}

\begin{lstlisting}[style=algo,title={Algorithm: Detection},label={lst:detection_stage}]
Algorithm DETECT(X, bins, green_list_percent, alpha=0.05, labels=None):
    n_check := num(X.columns)/2
    best := {p=+inf, z=None, seed_col=None, mark_col=None}

    if labels is provided:
        X := reorder_columns_by_feature_importance(X, labels)
    else:
        X := random_column_permutation(X)

    for each seed_col in columns(X):
        for each mark_col in columns(X), mark_col != seed_col:
            violations := 0

            for (s, m) in first_n_row_pairs(X, n_check, seed_col, mark_col):
                green := CREATE_GREEN_LIST(s, bins, green_list_percent)
                b := digitize(fractional_part(m), bins)
                if b not in green:
                    violations := violations + 1

            z := z_statistic(successes=n_check - violations, total=n_check)
            p := z_to_p(z)

            if (p <= alpha) and (z > 0):
                return True, {p, z, violations, seed_col, mark_col}

            if p < best.p:
                best := {p, z, violations, seed_col, mark_col}

    return False, best
\end{lstlisting}

\section{Spoofing algorithm}
\label{appendix:spoofing}
In this section, we provide the detailed spoofing \Cref{alg:fractional} described in \Cref{sec:experiment}. Given access to a watermarked dataset, either by the WGTD watermark~\citep{he2024watermarkinggenerativetabulardata} or by \Cref{alg:tabular} and without knowing the exact underlying parameters used for watermarking, an adversary attempt to create a synthetic dataset that can reliably be detected as `watermarked' by the corresponding detector. The goal of this `spoofing' attack is to then generate a dataset with toxic or harmful content that are detrimental to the original watermark provider's reputation. Watermark spoofing attacks have been studied extensively in both traditional relational database watermark, and more recently, in large language model watermarks. Our proposed algorithm~\ref{alg:fractional} presents a first attempt at studying the effects of spoofing attacks on generative tabular data that was missing from prior work in this line of research.  
\begin{algorithm}[ht]
\caption{Fractional  Replacement }\label{alg:fractional}
\begin{algorithmic}[1]
\renewcommand{\algorithmicrequire}{\textbf{Input:}}
\renewcommand{\algorithmicensure}{\textbf{Output:}}
\Require Synthetic dataset $\mathbf{S} \in \mathbb{R}^{m \times n}$, Watermarked dataset $\mathbf{W} \in \mathbb{R}^{m \times n}$.
\Ensure Modified dataset $\mathbf{X}'$.

\State Compute the fractional parts of all elements in $\mathbf{S}$ as $\mathbf{S\_fractions}$.
\State Compute the fractional parts of all elements in $\mathbf{W}$ as $\mathbf{W\_fractions}$.

\For{each element $x_{ij}$ in $\mathbf{S}$}
    \State Extract the integer part: $\floor{x_{ij}}$.
    \State Compute the fractional part: $\text{frac\_part} = |x_{ij} - \lfloor x_{ij} \rfloor|$.
    \State Find the closest fractional part: $\text{closest\_frac} = \arg\min_{f \in \mathbf{S\_fractions} \cup \mathbf{W\_fractions}} |f - \text{frac\_part}|$.
    \If{$x_{ij} < 0$}
        \State Replace $x_{ij}$ with $\floor{x_{ij}} - \text{closest\_frac}$.
    \Else
        \State Replace $x_{ij}$ with $\floor{x_{ij}} + \text{closest\_frac}$.
    \EndIf
\EndFor

\end{algorithmic}
\end{algorithm}

\section{Additional details on the algorithms in Section~\ref{sec:watermark}}
\paragraph{Note on \cref{alg:pair}.} There are several algorithms that creates a weighted permutation of $[2n]$ in \cref{algline:pair-permutation} of \cref{alg:pair}~\citep{fisher1953statistical, efraimidis2006weighted, braverman2015weighted, vieira2019sampling}. 

\paragraph{Runtime of \cref{alg:tabular}.} The majority of the computational bottleneck in \cref{alg:tabular} is from the two \texttt{for} loops that run over $keys$ and $values$ respectively. The \texttt{for} loop corresponding to each $value$ column requires to go through each of the values and move them to the nearest green interval or keep them unchanged. Finding the nearest green interval requires $O(\log b )$ time, and so this loop requires $O(m\log b)$ time. In fact, this is the dominating cost within each of the \texttt{for} loop corresponding to each $key$ column. Thus, the total time of running \cref{alg:tabular} is $O(mn\log b)$.

\paragraph{Parameters of \Cref{alg:tabular}.} In \Cref{alg:tabular}, the seeds are generated using a hash function, where the input is a list of the center of each cell in the key column (line 4). In \Cref{alg:pair}, the seed in line 1 can be chosen at will. We save the exact (key-pair) columns generated by \Cref{alg:pair}, and use it for detection with Bonferroni correction. If the detector does not have access to the original pairing of keys and values, we perform detection as follows: For each value column, check each individual key column to find a match by using the stored seed and calculate the number of elements in green list intervals defined by this seed. This matching process between key and value columns take at most $O(n^2)$, where $n$ is the number of pairs. We additionally note that when we have a sufficiently large number of bins and data points (rows), it is unlikely that two different columns generate the same seeds, as it is unlikely that the fractional parts for a given feature across all the data points lie in the exact same bins. Therefore, the matching process is likely to identify the key columns correctly.

In practice, one can look at the target dataset to be watermarked and decide the appropriate bin sizes needed. For example, if the dataset has already been truncated to only 2 decimal digits, choosing the number of bins to be $100$ (respectively, bin size of $0.01$) is a good starting point. The data owner can also decide the bin size by choosing their tolerance for downstream utility and robustness. That is, the fidelity constraint represents an upper bound on the bin size, while the robustness constraints represent different lower bounds on the bin size. In our robustness experiments, we use $b=1000$ bins for all evaluations.

\section{Proof of Fidelity}
\label{appendix:fidelity}
In this section, we provide the analysis for the fidelity property of our watermarking algorithm. At first glance, the theorem statements in \Cref{thm:fidelity} and \Cref{cor:wasserstein} are similar to that of~\citet{he2024watermarkinggenerativetabulardata}. However, there is a crucial main difference between our \Cref{alg:tabular} and their Algorithm 1 that leads to different analysis. Specifically, in \Cref{alg:tabular}, we do not force the red and green intervals to always be next to each other like in Algorithm 1 of~\citet{he2024watermarkinggenerativetabulardata}. Hence, when we promote a cell in red interval to the nearest green interval (Line 9-10 in \Cref{alg:tabular}), this `destination' green interval can be far away from the original red interval. Meanwhile, for~\citet{he2024watermarkinggenerativetabulardata}, the distance between any red and green intervals is always $1$. Hence, their analysis for the $\ell_\infty$ distance between the original and watermarked datasets is incompatible with our setting. In the following proofs of \Cref{thm:fidelity} and \Cref{cor:wasserstein}, we provide high-probability upper bounds on the distances, while the guarantees in~\citet{he2024watermarkinggenerativetabulardata} always hold (with probability $1$). 
\paragraph{Proofs of \Cref{thm:fidelity}}
\begin{proof}
Let $x$ be an element from the one of the value columns in the table $\mathbf{X}$, such that $x \in I_j$, where $I_j$ is the interval $\left[\frac{j-1}{b}, \frac{j}{b}\right]$ and the label of $I_j$ is red. Then, let $x_w$ be the corresponding watermarked value in $\mathbf{X}_w$. If $|x-x_w| > \frac{k}{b}$, then $x_w$ is either in one of the intervals $I_1 , \ldots, I_{\frac{j-k}{b}}$ or in one of the intervals $I_{\frac{j+k}{b}}, \ldots, I_{b}$. By \Cref{alg:tabular}, we know that there are no green intervals between $I_{\frac{j-k}{b}}$ and $I_{\frac{j+k}{b}}$. Since we label each interval as red or green independently with probability $1/2$, the probability that there are no green intervals between $I_{\frac{j-k}{b}}$ and $I_{\frac{j+k}{b}}$ is $(\frac{1}{2})^{2k-1}$, where we use the fact that there are $2k - 1$ intervals between $I_{\frac{j-k}{b}}$ and $I_{\frac{j+k}{b}}$ (excluding these). Then, we have
\begin{align*}
    \Pr\left(|x-x_w| > \frac{k}{b}\right) \leq \frac{1}{2^{2k - 1}}.
\end{align*}
Now using the union bound over all $x$ in the value columns of $\mathbf{X}$, we have 
\begin{align*}
    \Pr\left(\left\|\mathbf{X}-\mathbf{X}_w\right\|_{\infty} > \frac{k}{b}\right) \leq \frac{mn}{2^{2k - 1}}.
\end{align*}
Setting $\frac{mn}{2^{2k - 1}} \leq \delta$ and solving for $k$ gives us
\begin{align*}
    \Pr\left(\left\|\mathbf{X}-\mathbf{X}_w\right\|_{\infty} > \frac{1}{2b}\left(\log_2\left(\frac{mn}{\delta}\right) + 1\right)\right) \leq \delta.
\end{align*}
Since $\norm{X - X_w}_\infty \leq 1$ always, we have with probability at least $1 - \delta$ for $\delta \in (0,1)$, 
\begin{align*}
    \norm{\bX - \bX_w}_\infty \leq \min \left\{\frac{1}{2b}\left(\log_2\left(\frac{mn}{\delta}\right) + 1\right), 1\right \}.
\end{align*}

\end{proof}

\begin{remark}[Downstream processing]
    Since the input data $\bX$ and the watermarked data $\bX_w$ are close according to \Cref{thm:fidelity}, we can conclude that any downstream processing performed on these data sets will give similar results, as long as this processing can be implemented by a "nice" function.

    For example, if the downstream processing is described by a Lipschitz function $f\colon \mathbb{R}^{m \times 2 n} \to \mathbb{R}^{a \times b}$ with Lipschitz constant $L$, then $\norm{f(\bX) - f(\bX_w)}_r \leq L \norm{\bX - \bX_w}_d$ for some appropriate norms $\norm{\cdot}_d$ and $\norm{\cdot}_r$ on the domain and the range of $f$, respectively.
    Then, using equivalence of norms in finite dimensions, there is some constant $C$ (possibly dimension dependent) such that $\norm{\cdot}_d \leq C \norm{\cdot}_\infty$.
    Then, by \Cref{thm:fidelity}, we have
    \begin{equation}
        \norm{f(\bX) - f(\bX_w)}_r \leq L C \min \left\{\frac{1}{2b}\left(\log_2\left(\frac{mn}{\delta}\right) + 1\right), 1\right\}.
    \end{equation}
    Hence, for sufficiently large $b$, the output of such a downstream task acting on the input data and the watermarked data are going to be very similar.
\end{remark}

\paragraph{Proof of \Cref{cor:wasserstein}}
\begin{proof}
With probability $1 - \delta$, using \Cref{thm:fidelity} and definition of Wasserstein distance between empirical distributions, we have

\begin{align*}
    &\quad \cW_k(F_{\bX}, F_{\bX_w}) \\
    &\leq \left( \sum_{j=1}^m \frac{1}{m} \norm{\bX[j, :] - \bX_w[j, :]}_2^k \right)^{1/k}\\
    &\leq \left( \sum_{j=1}^m \frac{1}{m} \left( \sqrt{2n} \norm{\bX[j, :] - \bX_w[j, :]}_\infty \right)^k \right)^{1/k} \\
    &\leq \max_{j \in [m]}\left( \sqrt{2n} \norm{\bX[j, :] - \bX_w[j, :]}_\infty \right) \\
    &= \sqrt{2n}\|\bX - \bX_w\|_{\infty}\\
    &\leq \frac{\sqrt{n/2}}{b} \left(\log_2\left(\frac{mn}{\delta}\right) + 1\right).
\end{align*}

\end{proof}

\section{Proofs of Detectability}
\label{appendix:detectability}
In this section, we provide the analysis for the detection procedure described in \Cref{sec:dectection}. Similar to~\citet{he2024watermarkinggenerativetabulardata} and~\citet{zheng2024tabularmarkwatermarkingtabulardatasets}, we employ a hypothesis testing framework to detect our water in order to be robust to minor editing attacks. However, compared to~\citet{he2024watermarkinggenerativetabulardata}, we \textit{do not need} the assumptions that (i) the number of bins $b \rightarrow \infty$ and (ii) the data comes from a distribution with continuous probability density function for our analysis to hold. This implies in particular that our results here are non-asymptotic, distribution independent, and exact. Consequently, the proof structure for our \Cref{lem:convergence} is different from that of~\citet[Lemma 1]{he2024watermarkinggenerativetabulardata}. On the other hand, \citet{zheng2024tabularmarkwatermarkingtabulardatasets} does not provide theoretical analysis for the null distribution before the dataset is watermarked. 
\paragraph{Proof of \Cref{lem:convergence}}

\begin{proof}
    Let $I_j = \left[ \frac{j-1}{b}, \frac{j}{b} \right]$ for all $j\in [b]$ denote an interval on $[0,1]$. Define the indicator random variable $\bm{1}_g: [b] \to \{0,1\}$ as $\bm{1}_g(j)=1$ if $I_j$ is labeled green. Then given $\bm{1}_g$, we know which intervals are green and which intervals are red. Also, define $\cG = \cup_{j\in [b], \bm{1}_g(j)=1} I_j$ as the union of the intervals which are labeled green. So by definition $\cG$ is a set valued random variable.
    Therefore conditioned on $\bm{1}_g$ we have
    \begin{align}
        \Pr\left[x \in \cG \mid \bm{1}_g\right] &= \sum_{j \in [b], \bm{1}_g(j) = 1} \Pr[x \in I_j] \\
        &= \sum_{j=1}^b \Pr[x \in I_j] \cdot \1_g(j)
    \end{align}
    where the second equality uses the fact that $\1_g(j) \in \{0,1\}$ for all $j\in [b]$.
    Then, using the fact that the expected value of conditional probability of an event gives the probability of that event we have
    \begin{align}
        \Pr(x \in \cG) &= \E\left[\Pr\left[x \in \cG \mid \bm{1}_g\right]\right]\\
        &= \sum_{j=1}^b \Pr[x \in I_j] \cdot \E[\1_g(j)]\\
        &= \frac{1}{2}\sum_{j=1}^b \Pr[x \in I_j]\\
        &= \frac{1}{2},
    \end{align}
    where in the second equality we use that the expected value of $\1_g(j)$ is the probability that $I_j$ is labeled green (this is equal to $1/2$ by assumption), and in the last equality we use the fact that $x$ is supported in $[0,1]$.
\end{proof}
\section{Proofs of Robustness}
\label{appendix:robustness}
In this section, we provide the analysis for the results presented in \Cref{sec:robustness}. Compared to prior work of~\citet{he2024watermarkinggenerativetabulardata} and~\citet{zheng2024tabularmarkwatermarkingtabulardatasets}, we show theoretically that our watermarking approach is robust to both additive noise (addressed in prior work), truncation, and feature selection (novel contributions). In our analysis for robustness against additive noise, we first provide a meta-theorem (\Cref{thm:robustness}) that show our watermarking approach is robust to \emph{any} additive noise attack where the noise is sampled i.i.d. from an arbitrary distribution. The specific results for uniform noise (\Cref{thm:robustness-to-uniform-noise}) and Gaussian noise (\Cref{thm:robustness-to-gaussian-noise}) are presented as corollaries. Due to the structural differences in how we embed the watermark in \Cref{alg:fractional} compared to prior work, their robustness analysis (or lack thereof) does not apply to our setting. 

\paragraph{Proof of \cref{thm:truncation}}
\begin{proof}
    If $\xtr$ falls out of the original green interval $I_j$, then we know that either $\xtr > \nicefrac{j}{b}$ or $\xtr < \nicefrac{j-1}{b}$. Since $b \leq 10^p$, we know that the green interval $I_j$ lies in the union of at most two consecutive grids. We consider two cases depending on whether $I_j$ contains a grid point or not.  
    \paragraph{When $I_j$ does not contain a grid point.} 
    Since the truncation function defined in \cref{eq:trunc} always truncate an element $x$ to the nearest left grid point, which lies outside of $I_j$, we have 
    \begin{equation*}
        \Pr \left[\xtr \notin I_j | c \notin I_j, \forall c \in \text{grid} \right] = 1.
    \end{equation*}
    \paragraph{When $I_j$ contains a grid point.}
    Let $c$ denote the grid point in $I_j$. Then, when the interval $I_j$ contains a grid point, we have:
    \begin{align*}
        \Pr[\xtr \notin I_j | c \in I_j] &= \Pr \left[x \in \Big [ \frac{j-1}{b}, c \Big ) \right] \\
        &= \frac{c - \frac{j-1}{b}}{\frac{1}{b}} \\
        &= c \cdot b - j + 1
    \end{align*}
Then, summing over all possible events, we have the probability that the truncation attack successfully moves a watermarked element out of its original green interval is:
\begin{align*}
    &\quad \Pr[\xtr \notin I_j] \\
    &= \Pr[\xtr \notin I_j | c \notin I_j, \forall c \in \text{grid}] \cdot \Pr[c \notin I_j, \forall i \in \text{grid}] \\
    &+ \Pr[\xtr \notin I_j | c \in I_j] \cdot \Pr[c \in I_j]\\
    &= \Pr[c \notin I_j, \forall c \in \text{grid}] + (c \cdot b - j + 1) \Pr[c \in I_j] \\
    &= \left( \frac{b-1}{b} \right)^{10^p} + \frac{c \cdot b - j + 1}{b} \\ 
    &= \frac{(b-1)^{10^p} + b^{10^p - 1}(c \cdot b - j + 1)}{b^{10^p}} 
\end{align*}
Hence, the probability that $\xtr$ is in a red interval is $\rho_{\textrm{trunc}} = \frac{(b-1)^{10^p} + b^{10^p - 1}(c \cdot b - j + 1)}{2b^{10^p}}$. 

\end{proof}
\subsection{Robustness to additive noise}

\begin{theorem}[Robustness under noise]\label{thm:robustness}
    Fix a column in the watermarked dataset. Let $n_a$ be the number of cells in this column that an adversary can inject noise into. Let $S$ be a subset of $[b]$ such that $|S| = n_a$. Also let the noise injected into each cell is drawn i.i.d. from some distribution, \ie $\epsilon \sim \mathcal{D}$.
    
    Define parameter $\gamma$ as
    \begin{align}\label{eq:gamma}
        \gamma &\coloneqq \frac{1}{n_a}\sum_{i\in S}\sum_{j \in [b]} \Pr[(x_i+\eps_i)^\circ \in I_j],
    \end{align}
    where $\gamma \leq 1$ and $\epsilon_i$ is the noise added to $x_i$.
    Then, we must have
    \begin{align}
        n_a \geq \frac{m-z_{\textnormal{th}}\sqrt{m}}{2-\gamma},
    \end{align}
    for the expected $z$-score to be less than $z_{\textnormal{th}}$ (i.e., to remove the watermark).
\end{theorem}
\begin{proof}
    Define $x^\circ = x - \lfloor x \rfloor$ as the fractional part of any $x \in \R$.
    
    Let $I_j = \left[ \frac{j-1}{b}, \frac{j}{b} \right]$ for all $j\in [b]$ denote an interval on $[0,1]$. Define the indicator random variable $\bm{1}_g: [b] \to \{0,1\}$ as $\bm{1}_g(j)=1$ if $I_j$ is labeled green. Then given $\bm{1}_g$, we know which intervals are green and which intervals are red. Also define $\cG = \cup_{j\in [b], \bm{1}_g(j)=1} I_j$ as the union of the intervals which are labeled green. So by definition $\cG$ is a set valued random variable.

    Furthermore, we define $S$ to be a subset of row indices, $[m]$, in which we inject noise, such that $|S| = n_a$. For each $i \in S$, let $\epsilon_i$ be i.i.d. samples drawn from $\mathcal{D}$, and for each $i\notin S$, $\epsilon_i = 0$.

    We want to find a bound on $n_a$ such that in expectation the number of cells in a column that are watermarked is below the required $z$-score threshold. Given $z$-score threshold equal to $z_{\textnormal{th}}$, the minimum number of cells that need to be in the green intervals is $T_0 = \frac{z_{\textnormal{th}}\sqrt{m} + m}{2}$. The expected number of cells in the column where we add noise that is still in a green interval is
    \begin{align}
        \E \left[\sum_{i\in S} \bm{1}[(x_i + \eps_i)^\circ \in \cG] \mid \bm{1}_g\right] &= 
        \sum_{i\in S}\Pr\left[(x_i + \eps_i)^\circ \in \cG \mid \bm{1}_g\right].\label{eq:exp-on-x+eps-in-G}
    \end{align}
    For all $i\in S$, conditioned on $\bm{1}_g$ we have
    \begin{align}
        \Pr\left[(x_i+\epsilon_i)^\circ \in \cG \mid \bm{1}_g\right] &= \sum_{j \in [b], \bm{1}_g(j) = 1} \Pr[(x_i+\epsilon_i)^\circ \in I_j] \\
        &= \sum_{j=1}^b \Pr[(x_i + \epsilon_i)^\circ \in I_j] \cdot \1_g(j),\label{eq:x+eps-not-in-G}
    \end{align}
    where the second equality uses the fact that $\1_g(j) \in \{0,1\}$ for all $j\in [b]$.

    From \cref{eq:exp-on-x+eps-in-G} and \cref{eq:x+eps-not-in-G} and taking expectation over the labeling of the intervals we have 
    \begin{align}
        \E \left[\sum_{i\in S} \bm{1}[(x_i + \eps_i)^\circ \in \cG]\right] &= \E\left[\E \left[\sum_{i\in S} \bm{1}[(x_i + \eps_i)^\circ \in \cG] \mid \bm{1}_g\right]\right] \\
        &= 
        \sum_{i\in S}\sum_{j\in [b]}\Pr[(x_i + \epsilon_i)^\circ \in I_j] \cdot \E\left[\1_g(j)\right]\\
        &= \frac{n_a\gamma}{2},
    \end{align}
    where the last line uses the definition of $\gamma$ in \cref{eq:gamma}.

    Now we prove that $0 \leq \gamma \leq 1$. Since $\gamma$ is the sum of probabilities, it is non-negative. Since, 
    \begin{align}
        \gamma = \frac{1}{n_a}\sum_{i\in S}\sum_{j\in [b]} \Pr[(x_i+\eps_i)^\circ \in I_j],
    \end{align}
    and for every $i\in S$, there is exactly one $j \in [b]$ such that $(x_i+\eps_i)^\circ \in I_j$. So, we must have $\gamma \leq 1$.

    The expected number of cells which are in a green interval within the column is then:
    \begin{align}
        m - n_a + \frac{n_a\gamma}{2}. \label{eq:remaining-green}
    \end{align}
    So, to break the watermark \cref{eq:remaining-green} should be less than $T_0$, which gives
    \begin{align}
        m - n_a + \frac{n_a\gamma}{2} &\leq \frac{z_{\textnormal{th}}\sqrt{m} + m}{2}\\
        n_a &\geq \frac{m-z_{\textnormal{th}}\sqrt{m}}{2-\gamma}.
    \end{align}
\end{proof}
\begin{remark}[Distribution independent robustness]\label{rem:robustness-independent-of-distro}
    We observe that since $0\leq \gamma \leq 1$, the lower bound on $n_a$ can also be expressed independent of $\gamma$ as 
    \begin{align}
        n_a \geq \frac{1}{2}\left(m-z_{\textnormal{th}}\sqrt{m}\right).
    \end{align}
    Thus, even when the distribution of the noise is not known, one can have an estimate on the maximum number of cells in a column of the watermarked table that needs to be corrupted to remove the watermark.
\end{remark}
\begin{corollary}[Robustness under uniform noise]\label{thm:robustness-to-uniform-noise}
    Fix a column in the watermarked dataset. Let $n_a$ be the number of cells in this column that an adversary can inject noise into. Let $S$ be a subset of $[b]$ such that $|S| = n_a$. Also let the noise injected into each cell is drawn i.i.d  from an uniform distribution, \ie $\epsilon \sim \unif[-\sigma, \sigma]$ for $0 < \sigma \leq 1$.
    
    Define parameter $\gamma$ as
    \begin{align}\label{eq:gamma-uniform}
        \gamma &\coloneqq \frac{1}{2n_a\sigma}\sum_{i\in S} \sum_{j=1}^b \bigg(\max\left(0, \min\left(\frac{j}{b}, x_i +\sigma\right) - \max\left(\frac{j-1}{b}, x_i -\sigma\right)\right) \notag\\
        &\qquad + \max\left(0, \min\left(\frac{j}{b}-1, x_i +\sigma\right) - \max\left(\frac{j-1}{b}-1, x_i -\sigma\right)\right) \notag\\
        &\qquad + \max\left(0, 1+\min\left(\frac{j}{b}, x_i +\sigma\right) - \max\left(1+\frac{j-1}{b}, x_i -\sigma\right)\right)\bigg),
    \end{align}
    where $\gamma \leq 1$.
    Then, we must have
    \begin{align}
        n_a \geq \frac{m-z_{\textnormal{th}}\sqrt{m}}{2-\gamma},
    \end{align}
    for the expected $z$-score to be less than $z_{\textnormal{th}}$ (i.e., to remove the watermark).
\end{corollary}
\begin{proof}
    We just need to show \cref{eq:gamma-uniform} holds, since the remaining claims of the statement follow directly from \cref{thm:robustness}. 
    There are the following three possible cases for $(x_i+\eps_i)^\circ$ using $\eps_i \sim \unif[-\sigma, \sigma]$
    \begin{align}
        (x_i+\eps_i)^\circ =
        \begin{cases}
            1+x_i+\eps_i \quad & x_i+\eps_i \in [-1, 0)\\
            x_i+\eps_i \quad & x_i+\eps_i \in [0,1]\\
            x_i+\eps_i -1 \quad & x_i+\eps_i \in (1, 2].
        \end{cases}
    \end{align}
    Then using the union bound we have,
    \begin{align}
        \Pr[(x_i+\eps_i)^\circ \in I_j] &= \Pr\left[\frac{j-1}{b}-1 \leq x_i+\eps_i \leq \frac{j}{b}-1\right] + \Pr\left[\frac{j-1}{b} \leq x_i+\eps_i \leq \frac{j}{b}\right] +\notag\\
        &\qquad \Pr\left[1+\frac{j-1}{b} \leq x_i+\eps_i \leq 1+\frac{j}{b}\right].\label{eq:probabilities}
    \end{align}
    \paragraph{Equating \cref{eq:probabilities}.} Since $e_i \sim \unif[-\sigma, \sigma]$, then $x_i +e_i \sim \unif[x_i-\sigma, x_i + \sigma]$. Then we have
    \begin{align}\label{eq:xi+ei-overlap}
        \Pr\left[\frac{j-1}{b} \leq x_i+\eps_i \leq \frac{j}{b}\right] &= \frac{\max\left(0, \min\left(\frac{j}{b}, x_i +\sigma\right) - \max\left(\frac{j-1}{b}, x_i -\sigma\right)\right)}{2\sigma},
    \end{align}
    where $\max\left(0, \min\left(\frac{j}{b}, x_i +\sigma\right) - \max\left(\frac{j-1}{b}, x_i -\sigma\right)\right)$ is the overlap between the intervals $[x_i-\sigma, x_i + \sigma]$ and $\left[\frac{j-1}{b}, \frac{j}{b}\right]$. 

    Similarly, we have 
    \begin{align}
        \Pr\left[\frac{j-1}{b} -1 \leq x_i+\eps_i \leq \frac{j}{b}-1\right] &= \frac{\max\left(0, \min\left(\frac{j}{b}-1, x_i +\sigma\right) - \max\left(\frac{j-1}{b}-1, x_i -\sigma\right)\right)}{2\sigma}\label{eq:xi+ei+1-overlap}\\
        \Pr\left[1+\frac{j-1}{b} \leq x_i+\eps_i \leq 1+\frac{j}{b}\right] &= \frac{\max\left(0, 1+\min\left(\frac{j}{b}, x_i +\sigma\right) - \max\left(1+\frac{j-1}{b}, x_i -\sigma\right)\right)}{2\sigma}\label{eq:xi+ei-1-overlap}
    \end{align}
    \paragraph{Combining everything.} Plugging \cref{eq:xi+ei+1-overlap}, \cref{eq:xi+ei-1-overlap}, and \cref{eq:xi+ei-overlap} in \cref{eq:probabilities} and subsequently in \cref{eq:gamma} gives us the claim. As stated before, the rest of the proof follows from \cref{thm:robustness}.
\end{proof}

\begin{corollary}[Robustness under Gaussian noise]\label{thm:robustness-to-gaussian-noise}
    Fix a column in the watermarked dataset. Let $n_a$ be the number of cells in this column that an adversary can inject noise into. Let $S$ be a subset of $[b]$ such that $|S| = n_a$. Also let the noise injected into each cell is drawn i.i.d  from an Gaussian distribution, \ie $\epsilon \sim \mathcal{N}(0, \sigma^2)$ for $0 < \sigma \leq 1$.
    
    Define parameter $\gamma$ as
    \begin{align}\label{eq:gamma-gaussian}
        \gamma &\coloneqq \frac{1}{n_a}\sum_{i\in S} \sum_{j=1}^b \frac{1}{\sqrt{2\pi}\sigma}\int_{\frac{j-1}{b}}^{\frac{j}{b}}\sum_{k=-\infty}^{\infty} \exp\left(-\frac{(\theta/2\pi-x_i+2\pi k)^2}{2\sigma^2}\right)d\theta,
    \end{align}
    where $\gamma \leq 1$ and $\epsilon_i$ is the noise added to $x_i$.
    Then, we must have
    \begin{align}
        n_a \geq \frac{m-z_{\textnormal{th}}\sqrt{m}}{2-\gamma},
    \end{align}
    for the expected $z$-score to be less than $z_{\textnormal{th}}$ (i.e., to remove the watermark).
\end{corollary}
\begin{proof}
    Similar to \cref{thm:robustness-to-uniform-noise} we only need to show \cref{eq:gamma-gaussian} holds, since the rest of the proof follows from \cref{thm:robustness}. Observe that when $\eps_i \sim \cN(0, \sigma^2)$, $(x_i+\eps_i)^\circ$ can be expressed as:
    \begin{align}
        (x_i+\eps_i)^\circ &= x_i+\eps_i-\lfloor x_i+\eps_i\rfloor.
    \end{align}
    Observe that the random variable $(x_i+\eps_i)^\circ$ is always the fractional part of $x_i+\eps_i$, and $x_i+\eps_i \sim \cN(x_i, \sigma^2)$. The distribution of $(x_i+\eps_i)^\circ$ is exactly the wrapped normal distribution \cite[\S 2.2.6]{jammalamadaka2001topics} with the density function
    \begin{align}\label{eq:wrapped-normal-pdf}
        f(\theta; x_i, \sigma) = \frac{1}{\sqrt{2\pi}\sigma}\sum_{k=-\infty}^{\infty} \exp\left(-\frac{(\theta/2\pi-x_i+2\pi k)^2}{2\sigma^2}\right),
    \end{align}
    and is obtained by wrapping the normal distribution around the unit circle, and rescaling the support of the random variable to $[0,1]$. Then we have
    \begin{align}\label{eq:integral-wnd}
        \Pr\left[\frac{j-1}{b} \leq (x_i+\epsilon_i)^\circ \leq \frac{j}{b} \right] = \frac{1}{\sqrt{2\pi}\sigma}\int_{\frac{j-1}{b}}^{\frac{j}{b}}\sum_{k=-\infty}^{\infty} \exp\left(-\frac{(\theta/2\pi-x_i+2\pi k)^2}{2\sigma^2}\right)d\theta.
    \end{align}
    Since $\Pr\left[\frac{j-1}{b} \leq (x_i+\epsilon_i)^\circ \leq \frac{j}{b} \right] = \Pr\left[(x_i+\epsilon_i)^\circ \in I_j\right]$, using the definition of $\gamma$ we have our claim.
\end{proof}
\begin{remark}
    While \cref{eq:gamma-gaussian} looks ominous, it has been extensively studied in the literature. Importantly there are analytical expressions (that avoids the infinite sum, e.g., von Mises distribution closely matches the density function of wrapped normal distribution \citep{mardia2009directional}) that can approximate \cref{eq:wrapped-normal-pdf} well. Furthermore, one can also empirically approximate \cref{eq:gamma-gaussian} by first approximating the infinite sum using small values of $k$ \citep{kurz2014efficient} and then evaluating the integral.
\end{remark}
\subsection{Feature selection}\label{appendix:feature-selection}

\cref{alg:tabular} takes a black-box pairing subroutine $\mathrm{PAIR}$ as an input to determine the set of $(key, value)$ columns. We consider following two feature pairing schemes: 
\begin{enumerate}[label=\textbf{P.\arabic*}]
    \item \emph{Uniform}: features are paired uniformly at random, \label{item:uniform}
    \item \emph{Feature importance}: features are paired according to the feature importance ordering, where features with similar importance are paired. \label{item:FI}
\end{enumerate}
For \ref{item:FI}, without loss of generality, we assume that the columns of the original dataset are ordered in descending order of feature importance. Note that this reordering of features does not affect the uniform pairing scheme (\ref{item:uniform}) and only serves to simplify notation for~\ref{item:FI} in our analysis. Now consider the scenario that a data scientist pre-processes the watermarked dataset $\bX_w$ and keep the columns with largest importance scores. If the features are selected using weights scaling according to~\ref{item:FI}, we want to understand how many $(key, value)$ pairs that are watermarked due to \cref{alg:tabular} are preserved.

Before proceeding, we require some notations which we state first. Let $T = (t_1, \ldots, t_{2n})$ be an ordered list of columns of a tabular dataset. Let $p=(p_1, \ldots, p_{2n})$ be non-uniform sampling probability of each column. Without loss of generality, assume $p_1 \geq \cdots \geq p_{2n} > 0$. Let $u = (u_1, \ldots, u_{2n})$ such that for all $i \in [2n]$, $u_i = \frac{1}{2n}$. Let $k \in \mathbb{N}^{+}$, such that $k \leq [2n]$, and let $T_k = (t_1, \ldots, t_k)$. 

Let $\sigma_r$ be a permutation of $[2n]$ obtained using any probability vector $r \in \R^{2n}$. For $m=1,\ldots,n$ define random variables and their sum as
\begin{align}
    X_m \coloneqq \mathbbm{1}\{\sigma_{2m-1}, \sigma_{2m} \in T_k\}, \quad \text{and} \quad X \coloneqq \sum_{i=1}^n X_m.
\end{align}
$X_m$ checks whether consecutive elements in position $2m-1$ and $2m$ after permutation are in $T_k$, and so $X$ counts the number of times a $(key, value)$ pair is formed from the columns in $T_k$.
We begin by stating the following theorem on the conditional probability of constructing a pair from top-$k$ items given $s_k < |T_k|-1$ samples from the top-$k$ items have already been sampled.

\begin{lemma}\label{lem:robustness-prob-bound}
    Let $p \in \R^{2n}$ be a probability vector sorted in descending order and $u$ is uniform probability distribution on $[2n]$. Let $T_k = [k]$ and $\sum_{i\in [k]} p_i > \frac{k}{2n}$. Let $C_k$ be the samples which have been drawn that lie in $T_k$ (so $|C_k| \leq 2m \leq k-2 $), then if $p_i \geq 1/(2n)$ for all $i\in [k]$, for any $ 2< m < n$, 
    \begin{align}
        \Pr_p [X_{m+1} = 1\ \lvert\ C_k] \geq \Pr_u [X_{m+1} = 1\ \lvert\ C_k].
    \end{align}
\end{lemma}
\begin{proof}
    Let $S_m$ be the set of items that have not yet been sampled. By construction, $C_k \subset [n]\setminus S_m$. Then we have 
    \begin{align}
        \Pr_p [X_{m+1} = 1\ \lvert\ C_k] = \sum_{i\in T_k \setminus C_k} \frac{p_i}{\sum_{j\in S_m}p_j}\cdot\frac{ (\sum_{j\in T_k \setminus C_k} p_j - p_i)}{\sum_{j \in S_m}p_j -p_i}.\label{eq:pr-p-x-p1}
    \end{align}
    
    Let $Q_m = \sum_{j\in T_k \setminus C_k} p_j$, $R_m = \sum_{j\in S_m} p_j$, and $T_m = R_m -Q_m$. We have $Q_m - p_i \geq (|T_k \setminus C_k| - 1)p_k$ and $T_m \leq (|S_m| - |T_k\setminus C_k|)p_k$. Using this we have
    \begin{align}
        \frac{\sum_{j\in T_k \setminus C_k} p_j - p_i}{\sum_{j\in S_m}p_j -p_i} &= \frac{Q_m - p_i}{Q_m+T_m -p_i}\\
        &\geq \frac{(|T_k \setminus C_k| - 1)p_k}{(|T_k \setminus C_k| - 1)p_k + T_m}\\
        &\geq \frac{(|T_k \setminus C_k| - 1)p_k}{(|T_k \setminus C_k| - 1)p_k + (|S_m| - |T_k\setminus C_k|)p_k}\\
        &= \frac{|T_k \setminus C_k| - 1}{|S_m| - 1},\label{eq:pr-p-x-p2}
    \end{align}
    where in the first inequality we use: if $a\geq t$ then $a/(a+b) \geq t/(t+b)$ for a fixed $b \geq 0$. Then plugging this in \cref{eq:pr-p-x-p1} we have
    \begin{align}
        \Pr_p [X_{m+1} = 1\ \lvert\ C_k] &\geq \frac{|T_k \setminus C_k| - 1}{|S_m| - 1}\cdot\sum_{i\in T_k \setminus C_k} \frac{p_i}{\sum_{j\in S_m}p_j}\\
        &= \frac{|T_k \setminus C_k| - 1}{|S_m| - 1}\cdot \frac{Q_m}{Q_m+T_m} \\
        &\geq \frac{|T_k \setminus C_k| - 1}{|S_m| - 1}\cdot \frac{|T_k\setminus C_k|p_k}{|T_k\setminus C_k|p_k+T_m}\\
        &\geq \frac{|T_k \setminus C_k| - 1}{|S_m| - 1}\cdot \frac{|T_k\setminus C_k|p_k}{|T_k\setminus C_k|p_k+(S_m - |T_k\setminus C_k|)p_k} \\
        &= \frac{|T_k\setminus C_k|\cdot(|T_k\setminus C_k| - 1)}{|S_m|\cdot(|S_m| - 1)},\label{eq:pr-p-X_m-eq-1}
    \end{align}
    where in the third line we again use: if $a\geq t$ then $a/(a+b) \geq t/(t+b)$ for a fixed $b \geq 0$.

    For uniform probabilities, for all $i\in[2n]$ we have $u_i = \frac{1}{2n}$, and so we have
    \begin{align}
        \Pr_u [X_{m+1} &= 1\ \lvert\ C_k] = \sum_{i\in T_k \setminus C_k} \frac{u_i}{\sum_{j\in S_m}u_j}\cdot\frac{ (\sum_{j\in T_k \setminus C_k} u_j - u_i)}{\sum_{j \in S_m}u_j -u_i} = \sum_{i\in T_k\setminus C_k} \frac{1}{|S_m|}\cdot\frac{|T_k \setminus C_k| - 1}{|S_m|-1} \\
        &= \frac{|T_k\setminus C_k|\cdot(|T_k\setminus C_k| - 1)}{|S_m|\cdot(|S_m|-1)}.\label{eq:pr-u-X_m-eq-1}
    \end{align}
    Then combining \cref{eq:pr-p-X_m-eq-1} and \cref{eq:pr-u-X_m-eq-1} we have our claim.
\end{proof}

While \cref{lem:robustness-prob-bound} shows that the conditional probability that a pair from top $k$ is formed, at any given point given some of the items appeared prior to that round were from the top-$k$, when compared to uniform sampling, computing an expectation bound is harder due to sampling without replacement (SWOR).

An equivalent problem to consider for the expectation bound is the following:
\begin{problem}
    Let there be $2n$ distinct items ordered $T = \{t_1, \ldots, t_{2n}\}$ with probabilities $p_1 \geq \cdots \geq p_{2n} > 0$ and $\sum_{i=1}^{2n} p_i = 1$. Let $T_k = \{1, \ldots, k\}$ be the top-$k$ indices. Define $P_k = \sum_{i=1}^k p_i$, and $\|p_{T_k}\|_2^2 = \sum_{i=1}^k p_i^2$. We create a random permutation $\sigma$ of $[2n]$ by either:
    \begin{enumerate}[label=(\alph*)]
        \item \textbf{Uniform SWOR}: all permutations equally likely without replacement, or\label{item:PL-u}
        \item \textbf{Proportional SWOR}: sequential sampling without replacement. At each step one samples index $i$ with probability proportional to $p_i$ and remaining elements in $T$.\label{item:PL-p}
    \end{enumerate}
    For $m=1,\ldots,n$ define random variables and the sum of the random variables as
    \begin{align}\label{eq:Xis-and-X}
        X_m \coloneqq \mathbbm{1}\{\sigma_{2m-1}, \sigma_{2m} \in T_k\}, \quad \text{and} \quad X \coloneqq \sum_{i=1}^n X_m.
    \end{align}
    Compare $\E_p[X]$ (\cref{item:PL-p}) with $\E_u[X]$ (\cref{item:PL-u}).
\end{problem}

We have the following results on the expectations
\begin{theorem}[Uniform sampling baseline]\label{thm:unif-sampling-baseline}
    Under uniform sampling $\E_u[X] = \frac{k(k-1)}{2(2n-1)}$.
\end{theorem}
\begin{proof}
    There are $(2n)!$ permutations of $2n$ distinct items. Fixing $j \in [n]$, the number of permutations of $[2n]$ such that the items in positions $2j-1$ and $2j$ are in $T_k$ is $2{k\choose 2}(2n-2)!$, where the $(2n-2)!$ counts the number of permutations of the remaining $2n-2$ items. Using the fact uniform SWOR for $2n$ items is equivalent to uniformly choosing a permutation of $2n$ items the probability that sampling a permutation $\sigma$ such that $\sigma_{2j-1}, \sigma_{2j} \in T_k$ is 
    \begin{align}
        \Pr[\sigma_{2j-1}, \sigma_{2j} \in T_k] = \frac{2{k\choose 2}(2n-2)!}{(2n)!} = \frac{k(k-1)}{2n(2n-1)}.
    \end{align}
    From \cref{eq:Xis-and-X} we have that 
    \begin{align}
        \E_u[X] = \sum_{j=1}^n \Pr_u[\sigma_{2j-1}, \sigma_{2j} \in T_k] = \frac{k(k-1)}{2(2n-1)}.
    \end{align}
\end{proof}

\begin{theorem}\label{thm:FI-robustness-exp}
    $\E_p[X] \geq \E_u[X]$ if either of these conditions hold:
    \begin{enumerate}
        \item $P_k^2 - \|p_{T_k}\|_2^2 \geq \frac{k(k-1)}{2(2n-1)}$, \label{thm-item:1}
        \item $P_k^2 \geq \frac{k(k-1)}{2(2n-1)}\cdot\frac{1}{1-\frac{\eta^2}{k}}$, where $\max_{i\in [k]} p_i \leq \eta \mu$,\label{thm-item:2}
        \item $\frac{P_k^2}{k}\left(k-1 - \frac{(\rho-1)^2}{4\rho}\right) \geq \frac{k(k-1)}{2(2n-1)}$, where $\rho = \frac{\max_{i\in[k]} p_i}{\min_{i\in[k]} p_i}$,\label{thm-item:3}
        \item $P_k^2 \geq \frac{k(k-1)}{2(2n-1)}\cdot\frac{1}{1-\frac{e^{\delta}}{k}}$, where $\delta \geq \log k - H_2(w)$,\label{thm-item:4}
        \item $P_k^2 \geq \frac{k(k-1)}{2(2n-1)}\cdot\frac{1}{1-\frac{1}{N_{\operatorname{eff}}}}$, where $N_{\operatorname{eff}} = \frac{1}{\sum_{i\in [k]}w_i^2}$.\label{thm-item:5}
    \end{enumerate}
\end{theorem}
\begin{proof}
We split the proof for each of the items above in the following paragraphs:
    \paragraph{Proof for \cref{thm-item:1}.}For some rate vector $r = (r_1, \ldots, r_{2n})$ (one can set $r = p$ in case of \cref{item:PL-p} and $r=u$ in case of \cref{item:PL-u}), define
    \begin{align}
        R \coloneqq \sum_{j=1}^{2n} r_j, \quad \text{and} \quad Q_k \coloneqq \sum_{j \in T_k} r_j,
    \end{align}
    Define $X$ with respect to $r$ analogous to \cref{eq:Xis-and-X}. The probability that the first two items both lie in $T_k$ is:
    \begin{align}
        \Pr_r(X_1=1) = \sum_{i\in T_k} \frac{r_i(Q_k-r_i)}{R(R-r_i)}.
    \end{align}
    Let $r^{(m)}$ be the remaining rate vector before forming adjacent pairs in index $2m-1$ and $2m$, and let $T^{(m)}$ and $T_k^{(m)}$ be the remaining items and the remaining top-$k$ items at round $m$. To simplify notation, define
    \begin{align}
        S(r^{(m)}; T_k^{(m)}) \coloneqq \sum_{i\in T_k^{(m)}} \frac{r_i^{(m)}(Q_k^{(m)}-r_i^{(m)})}{R^{(m)}(R^{(m)}-r_i^{(m)})}.
    \end{align}
    Then we have
    \begin{align}
        \E_r[X] = \sum_{m=1}^{n} S(r^{(m)}; T_k^{(m)}).
    \end{align}
    Observe that for $r=p$ and $R=1$ we have $T_k^{(1)} = T_k$. Using $1-p_i \leq 1$ we have 
    \begin{align}
        \Pr_p[X_1=1] = S(p;T_k) = \sum_{i\in T_k} \frac{p_i(P_k-p_i)}{1-p_i} \geq \sum_{i\in T_k} p_i(P_k-p_i) = P_k^2 - \|p_{T_k}\|_2^2.
    \end{align}
    Using linearity of expectation and the fact that $X_i$'s are indicator random variables, we have
    \begin{align}
        \E_p[X] \geq S(p;T_k),
    \end{align}
    then if 
    \begin{align}
        P_k^2 - \|p_{T_k}\|_2^2 \geq \frac{k(k-1)}{2(2n-1)},
    \end{align}
    we have $\E_p[X] \geq \E_u[X]$.

    \medskip
    
    For the remainder of the proof define $\mu \coloneqq P_k/ k$, and normalize $p_{T_k}$ such that for all $i\in [k]$, $w_i = p_i / P_k$. So, we have $\sum_{i\in [k]} w_i = 1$, and $\|p_{T_k}\|_2^2 = P_k^2\sum_{i\in [k]}w_i^2$. 

    \paragraph{Proof for \cref{thm-item:2}.} $\max_{i\in [k]} p_i \leq \eta \mu$, is equivalent to the statement $\max_{i\in [k]} w_i \leq \frac{\eta}{k}$. Then we have $w_i^2 \leq \frac{\eta^2}{k^2}$ for all $i \in [k]$ and so $\sum_{i\in [k]} w_i^2 \leq \frac{\eta^2}{k}$, which gets us
    \begin{align}
        P_k^2 - \|p_{T_k}\|_2^2 \geq P_k^2(1- \frac{\eta^2}{k}). 
    \end{align}
    Thus, if $P_k^2 \geq \frac{k(k-1)}{2(2n-1)}\cdot\frac{1}{1- \frac{\eta^2}{k}}$, then $\E_p[X] \geq \E_u[X]$.
    
    \paragraph{Proof for \cref{thm-item:3}.} We have $\rho \geq 1$ by assumption. Let $\sigma = \frac{1}{k}\sum_{i\in T_k}(p_i-\mu)^2$ be the variance. Then 
    \begin{align}
        \sum_{i\in T_k} p_i^2 = k(\mu^2 + \sigma^2) = \frac{P_k^2}{k} + k\sigma^2.
    \end{align}
    So we have
    \begin{align}
        S(p;T_k) = (1-\frac{1}{k})P_k^2 - k\sigma^2.
    \end{align}
    Let $a = \min_{i\in T_k} p_i$ and $b = \max_{i\in T_k} p_i$. Using Bhatia-Davis inequality 
    \citep[Theorem 1]{bhatia2000better} 
    we have $\sigma^2 \leq (b-\mu)(\mu-a)$. We know $\rho = b/a$ and $a\leq \mu \leq b$. Then substituting $b=\rho a$ and maximizing $\sigma^2$ over $a$ for a fixed $\mu$ we have
    \begin{align}
        a^* = \frac{\mu(\rho + 1)}{2\rho}, \quad\ \quad \mu-a^* = \frac{\mu(\rho -1)}{2\rho}, \quad \text{and} \quad b^*-\mu = \frac{\mu(\rho -1)}{2},
    \end{align}
    and so
    \begin{align}
        \sigma^2 \leq \frac{\mu^2(\rho - 1)^2}{4\rho}.
    \end{align}
    Then we have
    \begin{align}
        S(p;T_k) = (1-\frac{1}{k})P_k^2 - k\sigma^2 \geq (1-\frac{1}{k})P_k^2 - \frac{k\mu^2(\rho - 1)^2}{4\rho} = \frac{P_k^2}{k}\left(k-1 - \frac{(\rho-1)^2}{4\rho}\right).
    \end{align}
    Then if $\frac{P_k^2}{k}\left(k-1 - \frac{(\rho-1)^2}{4\rho}\right) \geq \frac{k(k-1)}{2(2n-1)}$ we have $\E_p[X] \geq \E_u[X]$.

    \paragraph{Proof for \cref{thm-item:4}.} Define order $2$ Rényi entropy of the normalized top-$k$ weights $w$ as
    \begin{align}
        H_2(w) \coloneqq -\log \sum_{i\in[k]} w_i^2.
    \end{align}
    From the definition of $\delta$ we have
    \begin{align}
        \delta &\geq \log k + \log \sum_{i\in[k]} w_i^2 = \log\left(k \sum_{i\in [k]}w_i^2\right)\\
        \sum_{i\in [k]}w_i^2 &\leq \frac{e^\delta}{k}.
    \end{align}
    Then we have
    \begin{align}
        P_k^2 - \|p_{T_k}\|_2^2 \geq P_k^2\left(1 - \frac{e^\delta}{k}\right),
    \end{align}
    which then implies that if $P_k^2 \geq \frac{k(k-1)}{2(2n-1)}\cdot\frac{1}{1-e^\delta/k}$ we have $\E_p[X] \geq \E_u[X]$.

    \paragraph{Proof of \cref{thm-item:5}.} Note $N_{\operatorname{eff}}$ is the effective support of $p_{T_k}$. Observe $\frac{1}{\sum_{i\in [k]}w_i^2} = \frac{P_k^2}{\|p_{T_k}\|_2^2}$. This then implies that $P_k^2 - \|p_{T_k}\|_2^2 = P_k^2\left(1-\frac{1}{N_{\operatorname{eff}}}\right)$. So if $P_k^2 \geq \frac{k(k-1)}{2(2n-1)}\cdot\frac{1}{1-\frac{1}{N_{\operatorname{eff}}}}$ we have $\E_p[X] \geq \E_u[X]$. 
    
    Finally observe that $N_{\operatorname{eff}} \geq \frac{k}{\eta^2}$ for assumption in \cref{thm-item:2}, $N_{\operatorname{eff}} = O\left(\frac{k}{\rho}\right)$ for assumption in \cref{thm-item:3}, and $N_{\operatorname{eff}} \geq ke^{-\delta}$ for assumption in \cref{thm-item:4}.
\end{proof}

\begin{remark}
    \cref{thm:FI-robustness-exp} identifies regions in the function $p \mapsto \E_p[X]$ for any probability vector $p$ where $\E_p[X] \geq \E_u[X]$. In fact there are other regions in the function where $\E_p[X] \geq \E_u[X]$. For example, one can show that any $p$ can be constructed from $u$ by transferring probability mass from the bottom index to the top index sequentially. Naturally, one way to show $\E_p[X] \geq \E_u[X]$ would be to show that the corresponding potential function $f(p) = \sum_{i=1}^k \frac{p_i(S-p_i)}{1-p_i}$, where $S = \sum_{i\in [k]}p_i$, is monotonically increasing or that it is Schur convex. Unfortunately, this function $f$ is concave when considering transfers within the top-$k$ indices which breaks down the argument. We therefore leave it as an open problem to show that $\E_p[X] \geq \E_u[X]$ for any $p$.
\end{remark}
\section{Decoding algorithms}\label{sec:query_learning}

Consider the case that the continuous space in $(0,1]$ is uniformly discretized using bins of size $1/b$ \ie construct bins in $(0,1]$ as -- $\left\{ (0,\frac{1}{b}], (\frac{1}{b},\frac{2}{b}], \ldots, (\frac{b-1}{b}, 1]\right\}$. Then each of these bins are labeled $\{0,1\}$ with some probability $1-p_i$ and $p_i$ respectively. We save this data-structure. During query time, each entry of the input table/row is checked to see if the values fall into the bins labeled $1$. One can then compute $z$-scores based on the count of watermarked data identified in the input table. The output of the query function is $0$ or $1$ depending on whether the table is not watermarked or watermarked, respectively.

For a given input to the query function and the output watermark label, we formalize this problem as a learning problem as follows.

\paragraph{Assumptions.} We assume that the number of bins $b$ is known to the detection algorithm. We make such an assumption because we use a statistical learning algorithm for detection, and if $b$ is not known apriori, the VC dimension of the relevant hypothesis class becomes infinity, giving us vacuous sample complexity bounds. The algorithm generating the watermark fixes a probability $p \in (0, 1)$, which is hidden from the detection algorithm. The watermarking algorithm draws $b$ independent and identically distributed samples from the Bernoulli distribution with parameter $p$. Let $y_1, \ldots, y_b \in \{0, 1\}$ denote the observed values. For $k \in [b]$, we interpret $y_k$ to be the label of the interval (or bin) $\left(\frac{k - 1}{b}, \frac{k}{b}\right]$. 

\paragraph{Query function.} 
Let $\Pi\colon (0, 1] \to [b]$ denote the canonical projection onto the bins, defined as $\Pi(x) = k$ if $x \in \left(\frac{k - 1}{b}, \frac{k}{b}\right]$ for $k \in [b]$.
Clearly, $\Pi$ determines the index of the bin into which the input falls.
Given a tabular data $\Xb \in (0, 1]^{m \times n}$ as input, we denote $\Pi^{m \times n}\colon (0, 1]^{m \times n} \to [b]^{m \times n}$ to be the function that implements $\Pi$ entry wise on $\Xb$.
Let $F_z\colon [b]^{m \times n} \to \mathbb{R}$ be a $z$-score function (chosen based on the problem),
which first maps the index of the bin in each entry of the tabular data to the corresponding label of that bin, and then processes these labels in an appropriate fashion.

We define the \emph{query function} $Q: (0, 1]^{m \times n} \to \{0, 1\}$ given by $Q(\bm{x}) = \sgn(F_z(\Pi^{m \times n}(\bm{x})))$, where $\sgn\colon \mathbb{R} \to \{0, 1\}$ is the sign (or indicator) function defined as
\begin{equation}
   \sgn(x) = \begin{cases}
                 1 & \textnormal{ if } x > 0 \\
                 0 & \textnormal{ otherwise}.
             \end{cases}  \label{eqn:sign}
\end{equation}
The query function tells us whether ($1$) or not ($0$) the input data is watermarked.
Since the labels $y_k$ are not known to the detection algorithm, the query function is also not known.

\paragraph{Goal.} We want to approximate the query function up to a small prediction error with high probability.
Suppose that $\Xb$ is the data random variable, taking values in $(0, 1]^{m \times n}$.
We are given $M$ independent and identically distributed training samples $(\Xb_1, Q(\Xb_1)), \dotsc, (\Xb_M, Q(\Xb_M))$.
We wish to use these training samples to learn the function $Q$,
such that given $\epsilon \in (0, 1)$ and $\delta \in (0, 1)$, we have
\begin{equation}
    {\Pr_{\Xb_1, \dotsc, \Xb_M}}\left(\Pr_{\Xb'}\left(\mathcal{A}((\Xb_i)_{i = 1}^M, \Xb') \neq Q(\Xb')\right) \leq \epsilon\right) \geq 1 - \delta, \label{eqn:learning_guarantee}
\end{equation}
where $\mathcal{A}\colon ((0, 1]^{m \times n})^M \times (0, 1]^{m \times n} \to \{0, 1\}$ is an algorithm that takes the training data $\bv{X}_1, \dotsc, \bv{X}_M \in (0, 1]^{m \times n}$ as input and outputs the query function $\mathcal{A}((\bv{X}_i)_{i = 1}^M, \cdot)$.
Given training data $\bv{X}_1, \dotsc, \bv{X}_M \in (0, 1]^{m \times n}$, the quantity $\Pr_{\Xb'}\left(\mathcal{A}((\bv{X}_i)_{i = 1}^M, \Xb') \neq Q(\Xb')\right)$ is the expected value of the $0$-$1$ loss function (with respect to $\bv{X}'$).
Note that for Equation~\eqref{eqn:learning_guarantee} to hold, the number of samples $M$ depends on $\epsilon$,  $\delta$, and possibly other parameters in the problem.
$M$ can be interpreted as the number of queries required to learn the label assignments up to a small error with high probability.

Since the number of bins $b$ is known to the decoding algorithm, it suffices to consider query functions of the form $h = h' \circ \Pi^{m \times n}$, where $h'\colon [b]^{m \times n} \to \{0, 1\}$.
Furthermore, we suppose that the data is binned before implementing the learning algorithm.
Thus, we can write $\mathcal{A}(\cdot, \cdot) = \mathcal{A}'((\Pi^{m \times n})^M(\cdot), \Pi^{m \times n}(\cdot))$, where $\mathcal{A}'\colon ([b]^{m \times n})^M \times [b]^{m \times n} \to \{0, 1\}$.
Denote $\overline{\Xb} = \Pi^{m \times n} \circ \Xb$ to be the random variable taking values in $[b]^{m \times n}$, obtained by binning the entries of $\Xb$.
Then, Equation~\eqref{eqn:learning_guarantee} can be rewritten as
\begin{equation}
    \Pr_{\overline{\Xb}_1, \dotsc, \overline{\Xb}_M}\left(\Pr_{\overline{\Xb}'}(\cA'((\overline{\Xb}_i)_{i = 1}^M), \overline{\Xb}') \neq Q'(\overline{\Xb}')) \leq \epsilon\right) \geq 1 - \delta, \label{eqn:learning_guarantee_bindata}
\end{equation}
For this reason, it suffices to assume that the input data corresponds to the labels of the bins.

\subsection{VC dimension bounds for tabular data using query function of \cite{he2024watermarkinggenerativetabulardata}}
\label{appendix:guang_query}
In \citet{he2024watermarkinggenerativetabulardata}, the $z$-score function is of the form $F_z(\bv{X}) = \beta_0 + \sum_{j = 1}^n \beta_j T_j(\bv{X}) + \sum_{j = 1}^n \alpha_j T_j(\bv{X})^2$, where for all $i, j$, we have $\beta_i, \alpha_j \in \mathbb{R}$ and $T_j(\bv{X})$ is the Hamming weight of the label string corresponding to the $j$th column of $\bv{X}$.
As noted earlier in this section, it suffices to focus our attention to the case where the input is in $[b]^{m \times n}$, obtained by binning the original data.
Our goal is to embed all possible query functions into a neural network and use known bounds on the VC dimension for learning the true query function.
\begin{theorem}[Query complexity bound for querying \cite{he2024watermarkinggenerativetabulardata} with tabular data]\label{thm:query-cmplxty-he2024}
    Given the number of bins $b$, number of rows $m$ and columns $n$ of the query table, there is a neural network that can use
    \begin{equation}
        O\left(\frac{m n b \log(m n b) \log(1/\epsilon) + \log\left(1/\delta\right)}{\epsilon}\right)
    \end{equation}
    training data $($labeled data consisting of query table and whether or not the table is watermarked according to~\cite{he2024watermarkinggenerativetabulardata}$)$, to learn the function which determines whether an input table is watermarked, with error at most $\epsilon > 0$ with respect to $0$-$1$ loss $($see Equation~\eqref{eqn:learning_guarantee_bindata}$)$, and with probability greater than or equal to $1 - \delta$ over the training samples.
\end{theorem}
\begin{proof}
    We begin by embedding the labels $k \in [b]$ into vectors, so as to vectorize the inputs $[b]^{m \times n}$.
    To that end, associate the label $k \in [b]$ to the standard unit vector $\bv{e}_k \in \mathbb{R}^b$, where $\bv{e}_k$ is the vector with $1$ at the $k$th entry and zero elsewhere.
    Let $\Xb \in [b]^{m\times n}$ be the input data, and let it be mapped to the vector $\oplus_{j = 1}^n \oplus_{i = 1}^m \bv{e}_{{\Xb}_{ij}}$. Here, the columns of $\Xb$ are vertically stacked on top of each other after vectorizing.
    
    Let $\bv{w} \in \{0, 1\}^b$ denote a candidate vector of labels.
    Observe that $\ip{\bv{w}, \bv{e}_k} = \bv{w}_k$ for all $k \in [b]$.
    Furthermore, for $j \in [n]$, we have $\ip{\bv{w}^{\oplus m}, \oplus_{i = 1}^m \bv{e}_{\bv{X}_{ij}}} = \sum_{i = 1}^m \bv{w}_{{\Xb}_{ij}}$. If $\bv{w}$ is the correct labeling vector for the $j$th column, then $\sum_{i = 1}^m \bv{w}_{{\Xb}_{ij}} = T_j(\bv X)$.
    By definition, we have $T_j(\bv{X}) \in [0, m]$ for all $j \in [n]$ and all $\bv{X} \in [b]^{m \times n}$.
    
    Next, we show how to obtain $T_j(\bv{X})^2$ from $T_j(\bv{X})$.
    First, add another layer where we map $T_j(\bv{X})$ to $(m + 1) T_j(\bv{X})$, by choosing the weight to be $m + 1$.
    Next, choose the following piecewise polynomial activation function:
    \begin{equation}
        \psi(z) = \begin{cases}
                      0 &\text{if } z \leq 0 \\
                      z &\text{if } z \in (0, m] \\
                      z^2 &\text{if } z > m.
                  \end{cases}
    \end{equation}
    Then, since $T_j(\bv{X}) \in [0, m]$, we obtain $\psi(T_j(\bv{X})) = T_j(\bv{X})$.
    On the other hand, since $(m + 1) T_j(\bv{X}) \in \{0\} \cup [m + 1, m (m + 1)]$, we obtain $\psi((m + 1) T_j(\bv{X})) = (m + 1)^2 T_j(\bv{X})^2$.
    The constant factor of $(m + 1)^2$ can be absorbed into the weights in the next layer.
    
    This motivates us to propose the following neural network architecture for the problem of learning the query function (see Figure~\ref{fig:tabular_data_nn_architecture}).
    \begin{figure}[!ht]
        \centering
        \includegraphics[width=0.35\textwidth, trim={1cm 5cm 23cm 1cm}, clip]{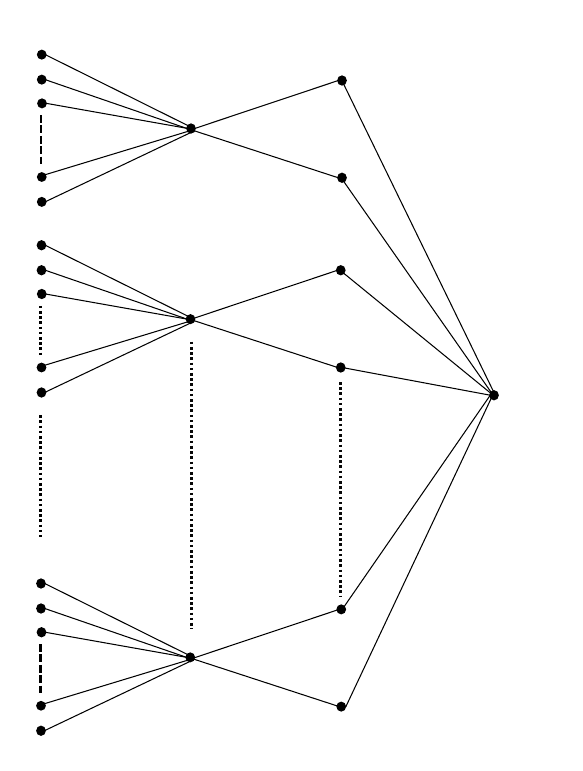}
        \caption{Neural network architecture to embed the problem of learning the query function. The activation function is applied at every node, except the nodes in the input and the output layers.}
        \label{fig:tabular_data_nn_architecture}
    \end{figure}
    The input nodes contain vectorized indices, such that columns of $\bv{X}$ are stacked on top of each other.
    Therefore, there are a total of $mnb$ input nodes.
    The first hidden layer computes $T_1(\bv{X}), \dotsc, T_n(\bv{X})$.
    Since the $T_j(\bv{X}) \in [0, m]$ for all $j \in [n]$, the activation function acting on the first hidden layer does nothing.
    The second layer effectively fans out $T_j(\bv{X})$ for $j \in [n]$. The top node maps $T_j(\bv{X})$ to itself, whereas the bottom node maps $T_j(\bv{X})$ to $(m + 1) T_j(\bv{X})$.
    Thus, applying the activation function at the second hidden layer gives us $T_j(\bv{X})$ and $(m + 1)^2 T_j(\bv{X})^2$ for $j \in [n]$.
    The output layer is obtained by applying weights and biases to the second hidden layer.
    Thus, the output of the neural network is of the form $F_z(\bv{X}) = \beta_0 + \sum_{j = 1}^n \beta_j T_j(\bv{X}) + \sum_{j = 1}^n \alpha_j T_j(\bv{X})^2$.
    Since we need to compute the VC dimension, we apply the sign (or indicator) function at the output layer.
    Therefore, the hypothesis class defined by this neural network is the set of functions computed by it over all the weights and biases at each layer, with a sign function (see~\eqref{eqn:sign}) applied at the end.
    By construction, all the query functions are a part of this hypothesis class.
    Thus, an upper bound on the VC dimension of the hypothesis class defined by the neural network also gives an upper bound on the VC dimension of the hypothesis class defined by all the query functions.
    
    There are a total of $O(m n b)$ weight and bias parameters before the first layer,
    $O(n)$ parameters before the second layer, and finally, $O(n)$ parameters before the output layer.
    Thus, we have a total of $\textnormal{wt} = O(m n b)$ parameters.
    We have a total of $L = 3$ layers.
    Let $\textnormal{wt}_i$ denote the number of weight and bias parameters from the input layer till the $i$th layer.
    Then, the effective depth of the neural network is $\overline{L} = \sum_{i = 1}^L \textnormal{wt}_i / \textnormal{wt}$ is of order $1$.
    Then, by~\cite[Theorem~6]{bartlett2019NNVCdim}, the VC dimension of the neural network hypothesis class is bounded above by $O(mnb \log(mnb))$.
    Therefore, by~\cite[Theorem~6.8]{shalev2014understanding}, we can infer that
    \begin{align}
        M = O\left(\frac{m n b \log(m n b) \log(1/\epsilon) + \log\left(1/\delta\right)}{\epsilon}\right)
    \end{align}
    training samples are sufficient to learn the query function to within an error of $\epsilon$ (with respect to the $0-1$ loss) with probability greater than or equal to $1 - \delta$ over the training samples.

    Now, denote $\cA'(\textnormal{training data}, \cdot)$ to be the query function that is learned from the training data.
    Then, if the distribution of the data is uniform, then $\Pr_{\overline{\Xb}'}(\cA'(\textnormal{training data}, \overline{\Xb}') \neq Q'(\overline{\Xb}'))$ is equal to the fraction of indices where the learned query function differs from the true query function.
    As a result, setting $\epsilon = 0.5/(m n b)$, we obtain the situation where we learn the query function exactly.
    Thus, we need $O((m n b)^2 \log(m n b)\log(b) + m n b \log(1/\delta))$ samples to learn the query function exactly with probability at least $1 - \delta$ over the training samples.
\end{proof}

\subsection{VC dimension bounds for tabular data using query function of Algorithm~\ref{alg:tabular}}
\label{appendix:tabular_bounds}

Unlike \cite{he2024watermarkinggenerativetabulardata}, in our work, we only watermark $n$ columns of the input table (see Algorithm~\ref{alg:tabular}). Thus, during query time, one would need to identify first the columns of the query table that are watermarked (referred to as value columns). As stated in Section~\ref{sec:dectection}, the query function returns $1$ if the $z$-score of each of the value columns are greater than an input threshold $z_{\text{th}} \geq 0$.
Thus, we prove the complexity of decoding~\Cref{alg:tabular} for $\alpha \leq 0.5$ (which corresponds to $z_{\textnormal{th}} \geq 0$).

\begin{theorem}[Query complexity bound for decoding~\Cref{alg:tabular}]
\label{thm:query-complexity-alg-tab}
    Given the number of bins $b$, number of rows $m$ and columns $2 n$ of the query table, and assuming that the significance level $\alpha$ for $z$-test in~\Cref{sec:dectection} is at most $0.5$, there is a neural network that can use
    \begin{equation}
        O\left(\frac{m n b \log(m n b) \log(1/\epsilon) + \log\left(1/\delta\right)}{\epsilon}\right)
    \end{equation}
    training data $($labeled data consisting of query table and whether or not the table is watermarked according to \Cref{alg:tabular}$)$, to learn the function which determines whether an input table is watermarked, with error at most $\epsilon > 0$ with respect to $0$-$1$ loss $($see Equation~\eqref{eqn:learning_guarantee_bindata}$)$, and with probability greater than or equal to $1 - \delta$ over the training samples.
\end{theorem}
\begin{proof}
    Following the proof of \Cref{thm:query-cmplxty-he2024}, we embed data $\bv{X} \in [b]^{m \times 2 n}$ before feeding into a neural network.
    First, we embed the labels $k \in [b]$ into vectors, so as to vectorize the inputs $[b]^{m \times 2 n}$.
    To that end, associate the label $k \in [b]$ to the standard unit vector $\bv{e}_k \in \mathbb{R}^b$, where $\bv{e}_k$ is the vector with $1$ at the $k$th entry and zero elsewhere.
    Let $\Xb \in [b]^{m \times 2n}$ be the input data, and let it be mapped to the vector $\oplus_{j = 1}^{2 n} \oplus_{i = 1}^m \bv{e}_{{\Xb}_{ij}}$. Here, each column of each row of $\Xb$ are vertically stacked on top of each other after vectorizing.

    Now, our goal is to show that there is a neural network architecture with appropriate weights, biases, and activation function such that the output of the neural network is the true query function.
    The true query function outputs $1$ \emph{if and only if} the $z$-score of all the \textit{value} columns is above the threshold $z_{\textnormal{th}}$.
    As per \Cref{sec:dectection}, we denote the $z$-score of the $j$-th column as $z_j = 2 (T_j(\Xb) / \sqrt{m}) - \sqrt{m}$, for all $j \in [2n]$.
    For obtaining the true query function, we need to know which columns correspond to the value columns as well as which bins are watermarked in a given value column.
    Since we seek to embed the true query function into a neural network, we assume that we know the indices $\mathcal{V}$ corresponding to the value column as well as the true labeling vectors.
    \Cref{alg:tabular}, by design, ensures that $\mathcal{V} \subseteq [2n]$ and $|\mathcal{V}| = n$.

    The first hidden layer contains a total of $2 n$ nodes.
    Let $\bv{w} \in \{0, 1\}^b$ denote a candidate vector of labels ($1$ is interpreted as watermarked, while $0$ is interpreted as not watermarked).
    Observe that $\ip{\bv{w}, \bv{e}_k} = \bv{w}_k$ for all $k \in [b]$.
    Furthermore, for all $j \in [n]$, we have $\ip{\bv{w}^{\oplus m}, \oplus_{i = 1}^m \bv{e}_{\bv{X}_{ij}}} = \sum_{i = 1}^m \bv{w}_{{\Xb}_{ij}}$.
    If $j \in \mathcal{V}$ and $\bv{w}$ is the correct labeling vector for the $j$th column, then $\sum_{i = 1}^m \bv{w}_{{\Xb}_{ij}} = T_j(\bv X)$.
    Since $T_j(\Xb) \in [0, m]$, we add a bias of $\sqrt{m} + 1$ so that $T_j(\Xb) + \sqrt{m} + 1 > \sqrt{m}$.
    If, on the other hand, if $j \notin \mathcal{V}$ , \ie we have a key column, then we set the weight $\bv{w}$ equal to the zero vector and add a bias of $m + \sqrt{m} + 2 > \sqrt{m}$.
    (We choose such a bias to internally distinguish a key column from a value column, since $T_j(\Xb) + \sqrt{m} + 1 < m + \sqrt{m} + 2$.)
    Subsequently, we apply the following piecewise-linear activation function:
    \begin{equation}
        \psi(x) = \begin{cases}
                      0 &\quad\textnormal{if } x \leq 0 \\
                      1 + \frac{1}{2 n} &\quad\textnormal{if } 0 \leq x \leq \sqrt{m} \\
                      x &\quad\textnormal{if } x > \sqrt{m}.
                  \end{cases}
        \label{eq:act-fn-alg-tab}
    \end{equation}
    Then, the output of the first hidden layer (after applying the weights, biases, and activation function) is equal to $T_j(\Xb) + \sqrt{m} + 1$ if $j \in \mathcal{V}$, while it is equal to $m + \sqrt{m} + 2$ if $j \notin \mathcal{V}$.
    
    For the second layer, our goal is to obtain the $z$-scores for the value columns and compare it to the threshold.
    Applying such a weight an bias to the value column $j \in \mathcal{V}$ gives $(2/\sqrt{m}) (T_j(\Xb) + \sqrt{m} + 1) - \sqrt{m} - 2 - 2/\sqrt{m} - z_{\textnormal{th}} = z_j - z_{\textnormal{th}}$, where $z_j$ is the $z$-score of the $j$th column.
    Now, if $z_j - z_{\textnormal{th}} \leq 0$ (or $z_j \leq z_{\textnormal{th}}$), then the activation function outputs $0$.
    On the other hand, if $z_j - z_{\textnormal{th}} > 0$, then since $z_j \leq \sqrt{m}$ and $z_{\textnormal{th}} \geq 0$ (as $\alpha \leq 0.5$ by assumption), we have $z_j - z_{\textnormal{th}} \leq \sqrt{m}$, so that the activation function outputs $1 + 1/(2n)$.
    Note that for the \emph{key} columns, we can set the weights as $1$ and biases as $0$. Then applying the activation function to these nodes, we have the output $m + \sqrt{m} + 2$ as before.

    Finally, in the third layer (which is also the final/output layer), for each column corresponding to the values, we apply a weight of $1$ and a bias of $-n$.
    For the key columns, we apply weight and bias equal to $0$ (which is equivalent to ignoring the key columns at the final layer).
    Therefore, the output at the final layer is $\sum_{j \in \mathcal{V}} (1 + 1/(2n)) \bm{1}[z_j - z_{\textnormal{th}} > 0] - n$, where $\bm{1}[A]$ denotes the indicator function of the event $A$ (i.e., $\bm{1}[A] = 1$ if $A$ is true and $0$ otherwise).
    It can be verified that the output is positive if and only if $z_j > z_{\textnormal{th}}$ for all $j \in \mathcal{V}$.
    Then, we apply the sign function given in \eqref{eqn:sign} to this output.
    Thus, such a network can learn the true query function required.
    A schematic of the proposed neural network architecture is given in \Cref{fig:tabular_data_alg1_nn_architecture}.

    \begin{figure}[!ht]
        \centering
        \includegraphics[width=0.35\textwidth, trim={1cm 5cm 23cm 1cm}, clip]{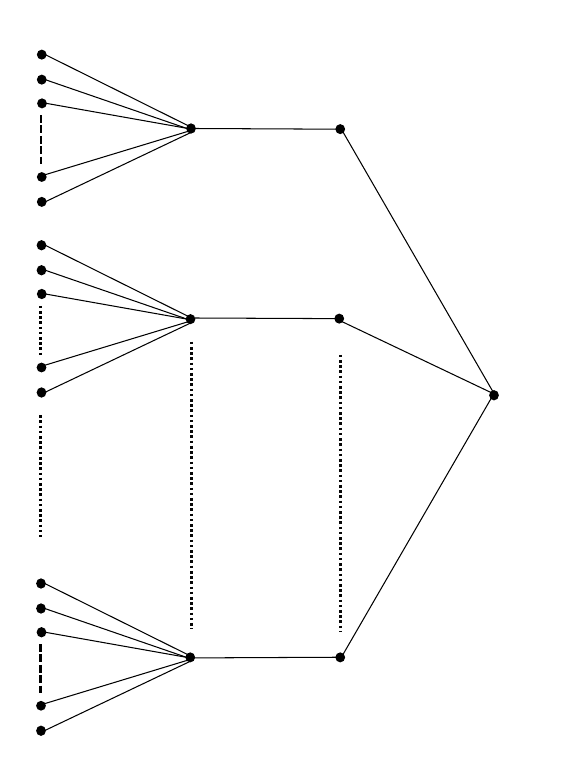}
        \caption{Neural network architecture to embed the problem of learning the query function for dataset watermarked according to \Cref{alg:tabular}. The activation function is applied at every node, except the nodes in the input and the output layers.}
        \label{fig:tabular_data_alg1_nn_architecture}
    \end{figure}

    The total number of parameters required to represent the weights and biases before the first layer is $O(mnb)$, before the second layer is $O(n)$, and before the third layer is $O(n)$. Thus we have a total of $O(mnb)$ parameters. We also have $3$ layers in the architecture. The effective depth of our neural network is $O(1)$. Then, using \cite[Theorem 6]{bartlett2019NNVCdim}, the VC dimension of the neural network hypothesis class is bounded above by $O(mnb\log(mnb))$. Using \cite[Theorem 6.8]{shalev2014understanding},
    \begin{align}
        M = O\left(\frac{mnb \log(mnb) \log(1/\epsilon) + \log(1/\delta)}{\epsilon}\right)
    \end{align}
    training samples are sufficient to learn the query function up to error $\epsilon$ with respect to the $0-1$-loss with probability greater than or equal to $1-\delta$ over the training samples.

    Now, denote $\cA'(\textnormal{training data}, \cdot)$ to be the query function that is learned from the training data.
    Then, if the distribution of the data is uniform, then $\Pr_{\overline{\Xb}'}(\cA'(\textnormal{training data}, \overline{\Xb}') \neq Q'(\overline{\Xb}'))$ is equal to the fraction of indices where the learned query function differs from the true query function.
    As a result, setting $\epsilon = 0.5/(m n b)$, we obtain the situation where we learn the query function exactly.
    Thus, we need $O((m n b)^2 \log(m n b)\log(b) + m n b \log(1/\delta))$ samples to learn the query function exactly with probability at least $1 - \delta$ over the training samples.
\end{proof}

\begin{corollary}[Query complexity of row queries to both \cite{he2024watermarkinggenerativetabulardata} and \Cref{alg:tabular}]
    For queries of the form $\Xb \in \R^{m\times 1}$ and query function that computes $z = 2T(\Xb)/\sqrt{m} - \sqrt{m}$ or any linear function of $T(\Xb)$, there is a neural network that can learn this query function up to error $\epsilon$ in $0$-$1$ loss using $O\left(\frac{mb\log(1/\epsilon) + \log(1/\delta)}{\epsilon}\right)$ samples with probability $1-\delta$ over the training samples.
\end{corollary}
\begin{proof}
    Using $n=1$ in \Cref{thm:query-complexity-alg-tab} gives us the result. For \cite{he2024watermarkinggenerativetabulardata}, observe that when a single row is input, $T_j(\Xb)$ for all $j \in n$ is in $\{0,1\}$. As such, the statistical test uses $z$-score tests on the standard normal distribution, and not the $\chi^2$ test. The hypothesis class corresponding to the neural architecture as defined in \Cref{thm:query-complexity-alg-tab} is thus PAC learnable for this problem, and so the query complexity bound follows by setting $n=1$.
\end{proof}

\subsection{A lower bound for row queries to \cite{he2024watermarkinggenerativetabulardata} given the $z$-scores}
\label{appendix:row_query}
 Consider the case that the continuous space in $(0,1]$ is uniformly discretized using bins of size $1/b$, \ie construct bins in $(0,1]$ as -- $\left\{ (0,\frac{1}{b}], (\frac{1}{b},\frac{2}{b}], \ldots, (\frac{b-1}{b}, 1]\right\}$. Then each of these bins are labeled $\{0,1\}$ with some probability $1-p$ and $p$ respectively. We save this data-structure, and allow anyone to query it with $S\in (0,1]^{m}$, responding with the $z$-score (Equation ~\ref{eq:z-score}) and whether the data is watermarked or not. In this section, we want to bound the minimum number of queries one can make to the model and estimate the watermarking scheme with high accuracy. Specifically we ask:
\begin{center}
\emph{
 When the number of bins $b$ is known, what is the query complexity to correctly identify all red and green intervals used for watermarking? }
\end{center}

Observe that for a row, the $z$-score is computed as: $2\sqrt{m}\left(\frac{T}{m} - \frac{1}{2}\right)$, where $T$ is the number of values in the row falling in the bins with label $1$.

When $b$ is known to the adversary, the problem can be reduced to the $2$-color Mastermind game. In $2$-color mastermind the codemaker creates a secret code using a sequence of colored pegs, and the codebreaker guesses the sequence. At each query by the codebreaker, the codemaker provides the number of pegs they correctly guessed. As such, given $z$ is the output of each query, one can easily compute the number of values in each bin $\{0,1\}$. Thus, given any $b$, one can reduce this problem to the $2$-color Mastermind game.
This specific form of Mastermind has been extensively studied in the literature \citep{chvatal1983mastermind, knuth1976computer}. In particular, with $2$ possible colors: red and green, the maximum number of row queries needed to recover the exact watermarking scheme is $\Theta\left( \frac{b \log(2)}{\log(b)} \right)$. Hence, the `red-green' watermarking scheme by \citet{he2024watermarkinggenerativetabulardata} can be learned by an adversary making $\Theta \left( \frac{b \log(2)}{\log(b)} \right)$ queries to the detector.
Note that this lower bound does not directly apply to the upper bounds derived using VC theory in \Cref{appendix:guang_query} and \Cref{appendix:tabular_bounds}, since here we assume access to the $z$-score for each query.
\end{document}